\documentclass{article}

\usepackage[english]{babel}

\usepackage[a4paper,top=2cm,bottom=2cm,left=3cm,right=3cm,marginparwidth=1.75cm]{geometry}

\usepackage[colorlinks=true, allcolors=blue]{hyperref}
\usepackage{times}

\usepackage{soul}
\usepackage{url}
\usepackage[utf8]{inputenc}
\usepackage[small]{caption}
\usepackage{graphicx}
\usepackage{amsmath,amssymb,amsthm}
\usepackage{amsfonts}
\usepackage{booktabs}
\usepackage{color}
\usepackage{xfrac}
\usepackage{cleveref}
\usepackage{dsfont}
\usepackage[ruled,linesnumbered,vlined]{algorithm2e}
\usepackage{pgfplots}						
\usepackage{tikz}
\pgfplotsset{compat=1.17}
\usetikzlibrary{patterns}
\usepackage{pgfplotstable}
    \pgfplotsset{
    name nodes near coords/.style={
        every node near coord/.append style={
            name=#1-\coordindex,
            alias=#1-last,
        },
    },
    name nodes near coords/.default=coordnode
    }

\newtheorem{theorem}{Theorem}[section]
\newtheorem{corollary}[theorem]{Corollary}
\newtheorem{lemma}[theorem]{Lemma}
\newtheorem{observation}[theorem]{Observation}
\newtheorem{example}[theorem]{Example}
\newtheorem{definition}[theorem]{Definition}
\newtheorem{proposition}[theorem]{Proposition}

\newenvironment{proofof}[1]{{\vspace*{5pt} \noindent\bf Proof of #1:  }}{\hfill\rule{2mm}{2mm}\vspace*{5pt}}

\makeatletter
\newtheorem*{rep@theorem}{\rep@title}
\newcommand{\newreptheorem}[2]{%
\newenvironment{rep#1}[1]{%
 \def\rep@title{#2 \ref{##1}}%
 \begin{rep@theorem}}%
 {\end{rep@theorem}}}
\makeatother

\newtheorem*{thm:three-agents}{Theorem~\ref{thm:three-agents}}
\newtheorem*{lemma:expanded-atom-path-general}{Lemma~\ref{lemma:expanded-atom-path-general}}
\newtheorem*{lemma:expanded-atom-path-k+1}{Lemma~\ref{lemma:expanded-atom-path-k+1}}
\newtheorem*{lemma:expanded-atom-path-k}{Lemma~\ref{lemma:expanded-atom-path-k}}
\newtheorem*{lemma:expanded-atom-path-k-1}{Lemma~\ref{lemma:expanded-atom-path-k-1}}
\newtheorem*{lemma:expanded-atom-path-k+1-goods}{Lemma~\ref{lemma:expanded-atom-path-k+1-goods}}

\newreptheorem{lemma}{Lemma}
\newreptheorem{claim}{Claim}

\newcommand{\bR}{{\mathbb{R}}}
\newcommand{\WPROP}{\mathsf{WPROP}}
\newcommand{\LL}{\mathsf{LL}}
\newcommand{\RR}{\mathsf{RR}}
\newcommand{\LR}{\mathsf{LR}}
\newcommand{\RL}{\mathsf{RL}}
\newcommand{\bX}{\mathbf{X}}
\newcommand{\bx}{\mathbf{x}}
\newcommand{\bs}{\mathbf{s}}
\newcommand{\bc}{\mathbf{c}}
\newcommand{\bv}{\mathbf{v}}
\newcommand{\bw}{\mathbf{w}}

\newcommand{\cI}{\mathcal{I}}

\newcommand{\splitting}{\textsf{Simple Splitting}\xspace}

\DeclareMathOperator*{\argmax}{argmax}
\DeclareMathOperator*{\argmin}{argmin}

\title{Tree Splitting Based Rounding Scheme for Weighted Proportional Allocations with Subsidy}

\author{
Xiaowei Wu\thanks{IOTSC, University of Macau. \{xiaoweiwu,yc17423\}@um.edu.mo. The authors are ordered alphabetically.}
\and Shengwei Zhou$^*$
}

\date{}

\begin{document}
\pagestyle{plain}
\maketitle

\begin{abstract}
We consider the problem of allocating $m$ indivisible items to a set of $n$ heterogeneous agents, aiming at computing a proportional allocation by introducing subsidy (money).
It has been shown by Wu et al. (WINE 2023) that when agents are unweighted a total subsidy of $n/4$ suffices (assuming that each item has value/cost at most $1$ to every agent) to ensure proportionality.
When agents have general weights, they proposed an algorithm that guarantees a weighted proportional allocation requiring a total subsidy of $(n-1)/2$, by rounding the fractional bid-and-take algorithm.
In this work, we revisit the problem and the fractional bid-and-take algorithm.
We show that by formulating the fractional allocation returned by the algorithm as a directed tree connecting the agents and splitting the tree into canonical components, there is a rounding scheme that requires a total subsidy of at most $n/3 - 1/6$. 
\end{abstract}
	
\section{Introduction}

We study the fair allocation problem of allocating a set of $m$ indivisible items $M$ to a group of $n$ heterogeneous agents $N$.
When all items $M$ give positive values to all agents, the problem refers to the fair allocation of \emph{goods}, which represents the situation of distributing resources or public goods.
Analogous but contrary, when all items $M$ give negative values (positive costs) to all agents, the problem refers to the fair allocation of \emph{chores}, which covers the task allocation among employees. 
In this paper, we mainly focus on the allocation of chores, where each agent $i$ has a cost function $c_i: 2^M \to \bR^+ \cup \{0\}$, but our main result also applies to the allocation of goods.
We consider the general weighted setting in which each agent $i$ has a $w_i > 0$ that represents her obligation to undertake the chores.
We normalize the weights of agents such that $\sum_{i\in N} w_i = 1$.
Traditionally, the fair allocation problem considers how to fairly allocate the items into $n$ bundles $(X_1, \ldots, X_n)$, while each agent receives exactly one bundle that guarantees some fairness criteria for her.
We say that agent $i$ has cost $c_i(S)$ for bundle $S \subseteq M$.
One natural and well-studied fairness notion is \emph{proportionality} (PROP)~\cite{steihaus1948problem}: an allocation is called weighted proportional (WPROP) if for all agent $i$, $c_i(X_i) \leq w_i \cdot c_i(M)$, where we refer to $w_i \cdot c_i(M)$ as the proportional share of agent $i$.

Unfortunately, when items are indivisible, PROP allocations are not guaranteed to exist, e.g., considering allocating a single item to two agents.
One possible way to circumvent this non-existence result is by introducing money to eliminate the inevitable unfairness.
Maskin~\cite{maskin1987fair} first proposed the setting called fair allocation with money that allows each agent to receive a subsidy $s_i \geq 0$ to eliminate unfairness, under which the objective is to minimize the total subsidized money.
Halpern and Shah~\cite{conf/sagt/HalpernS19} considered the fair allocation of goods with money, for another well-known fairness notion, called \emph{envy-freeness}~\cite{foley1967resource}.
Assuming that each item has a value of at most $1$ to each agent, they conjecture that a total subsidy of $n-1$ suffices to guarantee envy-freeness for the allocation of goods.
The conjecture was later verified by Brustle et al.~\cite{conf/sigecom/BrustleDNSV20}.

\paragraph{Allocation of Chores with Subsidy.}
The setting of fair allocation with subsidy can be naturally applied to the allocation of chore. 
Consider a situation where a factory, such as \emph{Foxconn}, wants to distribute a bunch of production tasks to $n$ different working units.
Different working units have different capacities, e.g., their proportional shares of the total tasks, and need to be compensated with money when assigned tasks exceeding their capacity, i.e., for overtime payment.
To ensure fairness (for a stable partnership), the mass of tasks each working unit receives, after removing the part that the subsidy offsets, should not exceed its proportional share of the total tasks.
On the other hand, to maximize benefits, the factory objects to minimize the total (additional) overtime payments, under the situation that all tasks are completed.
In other words, we aim to compute an allocation $(X_1,\ldots,X_n)$ and subsidies $(s_1,\ldots,s_n)$ such that $c_i(X_i) - s_i \leq w_i\cdot c_i(M)$ for all agent $i\in N$, with a small amount of total subsidy $\|\bs\|_1 = \sum_{i\in N} s_i$.
When all agents have equal weights, Wu et al.~\cite{conf/wine/WuZZ23} proposed an algorithm that computes a proportional allocation with a total subsidy of at most $n/4$.
Their algorithm is based on rounding a fractional allocation returned by the Moving Knife Algorithm.
However, when agents have general weights, their algorithm can only guarantee an upper bound of $(n-1)/2$ on the total subsidy.

\subsection{Our Results}
In this paper, we revisit the problem of fair allocation of chores with subsidy.
A commonly used framework for this problem is to first compute a fractional proportional allocation and then design a rounding scheme to turn the allocation integral with bounded total subsidy.
For the unweighted setting, the algorithm of Wu et al.~\cite{conf/wine/WuZZ23} that guarantees a total subsidy of at most $n/4$ is based on rounding a well-structured fractional allocation returned by the Moving Knife Algorithm.
Specifically, the algorithm computes a fractional allocation such that each agent receives a contiguous bundle of items in which at most two of them (the first and the last) are fractional.
However, the algorithm fails to generalize to the weighted setting where agents have arbitrary weights.
As they pointed out, any algorithm that requires each agent to receive a contiguous bundle of items, fails to compute a (fractional) weighted proportional allocation.

On the other hand, if we do not require any structural property on the fractional allocation, it is easy to compute a weighted proportional allocation.
For example, the Eating Algorithm that lets all agents eat their most preferred available item continuously at the same time, with speeds proportional to their weights, guarantees weighted envy-freeness, and thus weighted proportionality.
However, it is easy to find instances for which almost all items are shared by multiple agents in the fractional allocation returned by the Eating Algorithm. 
This complex item-sharing structure introduces major difficulties for devising the rounding scheme and bounding the incurred subsidy.
As a comparison, in the fractional allocation returned by the Moving Knife Algorithm, there are at most $n-1$ fractional items and each agent receives at most two such items.
Therefore, we consider the Fractional Bid-and-Take algorithm (FBAT) that computes fractional weighted proportional allocations with a certain level of canonical structure.
The same fractional algorithm is considered in~\cite{conf/wine/WuZZ23}, where it is shown that there are at most $n-1$ fractional items in the returned allocation.
However, in this fractional allocation, an agent may receive many fractional items.
Consequently, their rounding scheme can only guarantee an upper bound of $(n-1)/2$ on the total subsidy.
Whether strictly less subsidy is sufficient to guarantee (weighted) proportionality remains unknown and is proposed as an open problem.

In this paper, we revisit the fractional bid-and-take algorithm and answer this question affirmatively.
We characterize the structure of the fractional allocations returned by FBAT, by introducing the \emph{item-sharing} graph where agents are nodes and fractional items are edges.
We show that the item-sharing graph of the allocation returned by FBAT is a directed tree\footnote{As a comparison, the item-sharing graph defined by the allocation returned by the Moving Knife Algorithm is a path; while that for the Eating Algorithm can be any complex dense multi-graph.}.
We introduce a general rounding framework based on tree splitting, which allows us to split the directed tree graph into canonical components and round each component independently.
We show that there exists a rounding scheme that guarantees weighted proportionality with a total subsidy of at most $n/3-1/6$.

\medskip
\noindent
\textbf{Main Result.}
For the allocation of chores to a group of $n$ agents with general weights, we can compute in polynomial time a weighted proportional allocation with subsidy, where the total subsidy is at most $n/3-1/6$.

\smallskip

We remark that rounding the item-sharing graph via tree splitting requires striking a balance between the fineness of the splitting and tightness of the subsidy upper bound.
The finer we split the tree, the easier we can upper bound the total subsidy: but this upper bound is usually very loose.
For instance, if we take every edge as a component, we can easily recover the simple rounding scheme that guarantees a total subsidy of $(n-1)/2$~\cite{conf/wine/WuZZ23}.
On the other hand, regarding the tree as a whole when rounding items could help give tighter upper bounds but it would be very difficult to characterize the required subsidy.
One notable difficulty in designing the rounding scheme is to handle the fractional items that are shared by multiple agents, which we call the \emph{shattered items}.
A shatter item corresponds to multiple edges (that form a path, which we call an \emph{atom-path}) in the item-sharing graph.
Since every item should be rounded to exactly one agent, we need to ensure that all edges that correspond to the same shattered item belong to one component after splitting, which puts constraints on the feasible tree splittings.
To get around this difficulty, we design an atom-path-centric splitting that recursively splits the tree into an (near) atom-path and a collection of subtrees meeting certain constraints, and use mathematical induction to bound the total subsidy.

\smallskip

We further extend our tree splitting based rounding framework to the allocation of goods.
With some minor modifications to the analysis, we show that a total subsidy of at most $n/3$ suffices to guarantee a weighted proportional allocation of goods (the proof can be found in the appendix).


\subsection{Other Related Works}
Since the seminal work of Steinhaus~\cite{steihaus1948problem}, the fair allocation problem has received significant attention in the past decades.
Other than proportionality, another well-studied fairness notion is \emph{envy-freeness} (EF)~\cite{foley1967resource}, that is, no agent wants to exchange her bundle of items with another agent to improve her utility.
Since both EF and PROP are not guaranteed to exist when items are indivisible, a line of literature focused on the relaxations of these two fairness notions.
\emph{Envy-freeness up to one item} (EF1)~\cite{conf/sigecom/LiptonMMS04} and \emph{envy-freeness up to any item} (EFX)~\cite{journals/teco/CaragiannisKMPS19} are two widely studied relaxations of envy-freeness, which require that the envy between any two agents can be eliminated by removing some item; any item respectively.
Similarly, we have \emph{proportionality up to one item} (PROP1)~\cite{conf/sigecom/ConitzerF017,conf/aaai/BarmanK19} and \emph{proportionality up to any item} (PROPX)~\cite{journals/orl/AzizMS20}, which require the gap between the bundle's utility and the proportionality is bounded by some/any item, to be two relaxations of PROP.
In addition, \emph{maximin share} (MMS)~\cite{conf/bqgt/Budish10} is another popular relaxation of PROP.
For a more comprehensive overview of the fair allocation literature, we refer to the recent surveys by Amanatidis et al.~\cite{journals/ai/AmanatidisABFLMVW23} and Aziz et al.~\cite{journals/sigecom/AzizLMW22}.

\paragraph{Chore Allocation.}
Most of the fairness notions, such as EF1, EFX, PROP1, PROPX, and MMS, naturally extend to the analogous problem of chore allocation.
While the chore allocations are as common as the goods allocations in real-world scenarios, the problem receives less attention and usually tends to have worse results.
For the allocation of chores, both EF1~\cite{conf/approx/BhaskarSV21} and PROP1~\cite{conf/aaai/BarmanK19,conf/sigecom/ConitzerF017} allocations are guaranteed to exist and can be found in polynomial time.
In contrast to the goods allocations where PROPX allocation may not exist, PROPX allocations always exist for chores allocation~\cite{moulin2018fair,conf/www/0037L022}.
Regarding EFX allocations, they have only been shown to exist for some restricted settings~\cite{journals/aamas/AzizCIW22,gafni2023unified, conf/www/0037L022,journals/corr/abs-2308-12177,journals/ai/ZhouW24}.
While MMS allocations are not guaranteed to exist~\cite{conf/aaai/AzizRSW17,conf/wine/FeigeST21}, much attention focuses on approximate MMS allocations~\cite{aziz2022approximate,journals/teco/BarmanK20,conf/sigecom/HuangL21}, which led to the state-of-the-art ratio of $13/11$~\cite{conf/sigecom/HuangS23}.

\paragraph{Weighted Setting.}
Motivated by real-world applications where agents are usually not equally obliged, Chakraborty et al.~\cite{journals/teco/ChakrabortyISZ21} proposed the \emph{weighted} (or \emph{asymmetric}) setting.
Chakraborty et al.~\cite{journals/teco/ChakrabortyISZ21} introduced the \emph{weighted envy-freeness up to one item} (WEF1) for the allocation of goods and show that WEF1 allocations always exist and can be computed in polynomial time.
Lately, WPROP1 allocations have been proved to exist for chores~\cite{branzei2023algorithms}, and the mixture of goods and chores~\cite{journals/orl/AzizMS20}.
Li et al.~\cite{conf/www/0037L022} showed the existence and computation of WPROPX allocations of chores.
Recently, Wu et al.~\cite{conf/sigecom/0001Z023} and Springer et al.~\cite{conf/aaai/SpringerHY24} proposed algorithms for the computation of WEF1 allocations for chores.
The weighted version of MMS (WMMS) has also been studied for both goods~\cite{journals/jair/FarhadiGHLPSSY19} and chores~\cite{conf/ijcai/0001C019}.
In addition, generalizing MMS to the weighted setting also inspires the fairness notion of \emph{Anyprice Share} (APS).
Babioff et al.~\cite{conf/sigecom/BabaioffEF21} showed that there always exist $(3/5)$-approximate APS allocations for goods and $2$-approximate APS allocations for chores.
The approximate ratio for chores was recently improved to $1.733$ by Feige and Huang~\cite{conf/sigecom/FeigeH23}.

\paragraph{Fair allocation with Money.}
Beyond additive valuation functions, Brustle et al.~\cite{conf/sigecom/BrustleDNSV20} showed that a subsidy of $2(n-1)$ dollars per agent suffices to guarantee envy-freeness for general monotonic valuation functions.
Very recently, the upper bound was improved to $n-1$ dollars per agent by Kawase et al.~\cite{conf/aaai/KawaseMSTY24} while the total subsidy is at most $n(n-1)/2$.
Barman et al.~\cite{conf/ijcai/BarmanKNS22} considered the dichotomous valuations (each good has binary marginal value) and showed that envy-freeness can be guaranteed with a per-agent subsidy of at most $1$.
Regarding truthfulness, Goko et al.~\cite{goko2024fair} showed that for general monotone submodular valuations, there exists a truthful mechanism that guarantees envy-freeness with a subsidy of at most one dollar per agent.
A similar setting is called fair allocation with monetary transfers, introduced by Aziz~\cite{conf/aaai/000121}, that allows agents to transfer money to each other.
Instead of minimizing the total subsidy, their result focused on the characterization of allocations that are equitable and envy-free with monetary transfers.
When considering money as a divisible good, the setting of fair allocation with money is similar to the fair allocation of mixed divisible and indivisible items, which also receives much attention~\cite{journals/ai/BeiLLLL21,conf/approx/BhaskarSV21,conf/ijcai/LiLLT23}.
For a more detailed review of the existing works on mixed fair allocation, please refer to the recent survey~\cite{journals/corr/abs-2306-09564}.

\section{Preliminary} \label{sec:prelim}

In the following, we introduce the notation and the fairness notions.
We consider the problem of allocating $m$ indivisible chores $M$ to $n$ agents $N$ where each agent $i\in N$ has an additive cost function $c_i:2^M \to \bR^+ \cup \{0\}$.
We assume that the agents are weighted, i.e., each agent $i\in N$ has a weight $w_i > 0$ (that represents her obligation for undertaking the chores) and $\sum_{i\in N} w_i = 1$. 
We denote by $\bw = (w_1, \ldots, w_n)$ the weights of agents.
A cost function $c_i$ is said to be {\em additive} if for any bundle $S \subseteq M$ we have $c_i(S) = \sum_{e\in S} c_i(\{e\})$.
For convenience, we use $c_i(e)$ to denote $c_i(\{e\})$.
We use $\bc = (c_1, . . . , c_n)$ to denote the cost functions of agents.
We assume w.l.o.g. that each item has cost at most one to each agent, i.e. $c_i(e) \leq 1$ for any $i\in N$, $e\in M$.
Given an instance $\cI = (M,N,\bw,\bc)$, an allocation of $\cI$ is represented by an $n$-partition $\bX = (X_1,\ldots,X_n)$ of the items $M$, where $X_i \cap X_j = \emptyset$ for all $i \neq j$ and $\cup_{i\in N} X_i = M$.
In allocation $\bX$, agent $i\in N$ receives bundle $X_i$.
For convenience of notation, given any set $X\subseteq M$ and $e\in M$, we use $X+e$ and $X-e$ to denote $X\cup\{e\}$ and $X\setminus\{e\}$, respectively.

\begin{definition}[WPROP]
    An allocation $\bX$ is called weighted proportional (WPROP) if $c_i(X_i) \leq w_i \cdot c_i(M)$ for all $i\in N$.
\end{definition}

We use $\WPROP_i$ to denote agent $i$'s proportional share, i.e., $\WPROP_i = w_i \cdot c_i(M)$.

We use $s_i \geq 0$ to denote the subsidy we give to agent $i\in N$, $\bs = (s_1, \ldots, s_n)$ to denote the set of subsidies, and $\|\bs\|_1 = \sum_{i\in N} s_i$ to denote the total subsidy.

\begin{definition}[WPROPS] \label{def:PROPS}
    An allocation $\bX$ with subsidies $\bs = (s_1, \ldots, s_n)$ is called weighted proportional with subsidies (WPROPS) if for any $i\in N$ we have $c_i(X_i) - s_i \leq \WPROP_i$.
\end{definition}

Given any instance $\cI = (M, N, \bw, \bc)$, we aim to find WPROPS allocation $\bX$ with a small amount of total subsidy.
Unlike envy-freeness with subsidy\footnote{An allocation with subsidies is weighted envy-free with subsidy if $c_i(X_i)/w_i - s_i \leq c_i(X_j)/w_j - s_j$ for all $i,j\in N$. We remark that even in the unweighted setting, not all allocations are envy-freeable. In other words, there exist allocations that cannot guarantee envy-freeness even if we introduce unbounded subsidies to the agents.}, we can turn any allocation $\bX$ into weighted proportional by introducing subsidies to the agents as follows:
\begin{equation*}
    s_i = \max \{c_i(X_i) - \WPROP_i, 0\}, \qquad \forall i\in N.
\end{equation*}

Therefore, in the rest of this paper, we mainly focus on computing the allocation $\bX$. The subsidy to each agent will be automatically decided by the above equation.
Before we dive into our algorithm and analysis, we first introduce a reduction from general instances to identical ordering instances.
An instance is called \emph{identical ordering} (IDO) if all agents have the same ordinal preference over all items.

\begin{definition}[Identical Ordering (IDO) Instances]
    An instance is called identical ordering (IDO) if all agents have the same ordinal preference on the items, i.e., $c_i(e_1) \geq c_i(e_2) \geq \cdots \geq c_i(e_m)$ for all $i\in N$.
\end{definition}

In the unweighted setting, Wu et al.~\cite{conf/wine/WuZZ23} showed that if there exists an algorithm for computing PROPS allocations for IDO instances, then it can be converted to an algorithm that works for general instances while preserving the subsidy requirement.
The IDO reduction is a general technique and is widely used in the computation of approximate MMS allocations~\cite{journals/teco/BarmanK20,journals/aamas/BouveretL16,conf/sigecom/HuangL21,conf/sigecom/HuangS23}, PROP1/PROPX allocations~\cite{journals/orl/AzizMS20,conf/www/0037L022}.
We show that the reduction in~\cite{conf/wine/WuZZ23} can be naturally generalized to the weighted setting.
With this reduction, in the rest of this paper, we only consider IDO instances.
For completeness, we repeat the proof in Appendix~\ref{sec:missing-proofs}.

\begin{lemma} \label{lemma:reduction-to-IDO}
    If there exists a polynomial time algorithm that given any IDO instance computes a WPROPS allocation with total subsidy at most $\alpha$, then there exists a polynomial time algorithm that given any instance computes a WPROPS allocation with total subsidy at most $\alpha$.
\end{lemma}

\section{Technical Overview}

Similar to previous results on the unweighted setting, our algorithm has two main steps: we first compute a fractional WPROP allocation, in which a small number of items are fractionally allocated; then we find a way to round the fractional allocation to an integral one.
Since some agents may have costs exceeding their proportional share after rounding, we offer subsidies to these agents.

For the unweighted setting, it has been shown that there exists a rounding scheme with a total subsidy of no more than $n/4$, based on the fractional allocation returned by the Moving Knife Algorithm.
However, as pointed out by Wu et al.~\cite{conf/wine/WuZZ23}, the algorithm fails to compute fractional WPROP allocations when agents have different weights.\footnote{In fact, any algorithm that tries to allocate each agent a continuous interval (as the Moving Knife Algorithm computes) fails in the weighted setting. See~\cite{conf/wine/WuZZ23} for a counter-example.}
Hence, generalizing the results for the unweighted setting to the weighted setting is non-trivial.
For the weighted setting, they consider the Fractional Bid-and-Take Algorithm (FBTA) and show that there exists a rounding scheme with a total subsidy of no more than $(n-1)/2$.
Whether the fractional WPROP allocation returned by FBTA allows total subsidy strictly smaller than $(n-1)/2$ remains unknown.
In this paper, we revisit the algorithm and give a better characterization of the fractional WPROP allocation returned by FBTA.
We show that the returned fractional allocation forms a tree structure connecting the agents and turning the allocation integral requires a rounding scheme on the tree.
This provides a connection between the fair allocation problem with the graph theory.
Before our study, the more extensive connection with graph theory mainly focused on comparison-based fair allocations (when agents are regarded as nodes and the envies between agents are regarded as edges), e.g., ~\emph{envy graph} and~\emph{champion graph}~\cite{conf/sigecom/ChaudhuryGMMM21}.
We believe that this special tree structure, naturally decided by the fractional allocation, might be of independent interest.

\subsection{Fractional Bid-and-Take Algorithm}\label{ssec:FBTA}

We first introduce the notation for representing a fractional allocation.
We use $x_i(e) \in [0,1]$ to denote the fraction of item $e$ allocated to agent $i\in N$.
We use $\bx_i = \{x_i(e)\}_{e\in M} \in [0,1]^{m}$ to denote the fractional bundle allocated to agent $i\in N$ and $\bx = (\bx_1,\bx_2,\ldots,\bx_n)$ to denote the fractional allocation.
Note that the allocation is complete if and only if $\sum_{i\in N} x_i(e) = 1 $ for every $e \in M$.
In the fractional allocation $\bx$, we have $c_i(\bx_i) = \sum_{e \in M} (x_i(e) \cdot c_i(e))$ for each agent $i\in N$.
We call $\bx$ an integral allocation if $x_i(e) \in \{0,1\}$ for every $i\in N$ and $e\in M$.
To avoid ambiguity, we consistently use $\bX$ to refer to an integral allocation.

\paragraph{The Algorithm.}
We consider the IDO instance\footnote{The IDO input assumption here is not necessary. For any general instance, the Fractional Bid-and-Take algorithm computes a fractional WPROP allocation (see~\cite{conf/wine/WuZZ23}). Here we fix the instance to be IDO for ease of analysis for the later rounding part.} with $c_i(e_1) \leq c_i(e_2) \leq \cdots \leq c_i(e_m)$ and allocate the items one by one in a continuous manner.
We initialize all agents to be active.
For each item $e_j\in M$, we continuously allocate $e_j$ to the active agent $i$ with the minimum $\frac{c_i(e_j)}{c_i(M)}$\footnote{Recall that in this paper we normalize cost functions such that $c_i(e) \leq 1$ for all $i\in N, e\in M$.
When applying the different normalization s.t. $c_i(M) = 1$ for all $i\in N$, the algorithm corresponds to greedily allocating each item to the agent with minimum cost on it, which is consistent with the Bid-and-Take algorithm in~\cite{conf/www/0037L022}.}, until either $e_j$ is fully allocated or agent $i$ reaches her proportional share, i.e., $c_i(\bx_i) = \WPROP_i$. 
Besides, once an agent reaches her proportional share, we immediately inactivate her.
The algorithm terminates when all items are fully allocated.
We summarize the detailed algorithm in Algorithm~\ref{alg:FBTA}.

\begin{algorithm}[htbp]
    \caption{Fractional Bid and Take Algorithm (FBTA)}
    \label{alg:FBTA}
    \KwIn{An IDO instance $(M,N,\bw,\bc)$ with $c_i(e_1) \leq c_i(e_2) \leq \cdots \leq c_i(e_m)$ for all $i\in N$}
    $X_i \gets \mathbf{0}^m, \forall i\in N$  \qquad \qquad \tcp{current fractional bundle}
    $A\gets N$ \qquad \qquad  \qquad \qquad \tcp{ the set of active agents}
    $\mathbf{z} \gets \mathbf{1}^m$ \qquad \qquad  \qquad \qquad\tcp{remaining fraction of the items}
    $j\gets 1$ \qquad \qquad  \qquad \qquad \> \> \tcp{item to be allocated}
    \While{$j \le m$}{
        Let $i \gets \argmin_{i'\in A} \frac{c_{i'}(e_j)}{c_{i'}(M)}$; \qquad \qquad  \qquad \qquad \tcp{break tie lexicographically}\
        \If{$c_i(\bx_i) +z_j\cdot c_i(e_j)>\WPROP_i$}{
            $x_i(e_j) \gets \frac{\WPROP_i-c_i(\bx_i)}{c_i(e_j)}$; \qquad\qquad\qquad\quad \tcp{it is guaranteed that $x_i(e_j)\leq 1$} 
            $z_j \gets z_j - x_i(e_j)$, $A \gets A\setminus \{i\}$\;
        }
        \Else{
            $x_i(e_j) \gets z_j$, $z_j \gets 0$, $j\gets j+1$\;
        }    
    }
    \KwOut{A fractional allocation $\bx  = (\bx_1,\ldots, \bx_n)$.}
\end{algorithm}

\begin{lemma}[\cite{conf/wine/WuZZ23}]\label{lemma: wprop}
    The output allocation $\bx$ is a fractional WPROP allocation. 
\end{lemma}


Note that in the above WPROP allocation, the number of fractional items is at most $n-1$, since the number of fractional items can increase only when some agent becomes inactive, and there is at least one active agent when the algorithm terminates.
However, in the allocation, an agent might receive many fractional items, which introduces difficulties in designing the rounding scheme.

Toward a better understanding of the FBTA, we give an example instance of the chores allocations to weighted agents.
Throughout the paper, whenever we use $\cI^*$, we mean the instance described as follows.
\begin{example} \label{example:fractional-allocation}
    Consider the following instance $\mathcal{I}^*$ given by $(M, N, \bw, \bc)$.
    The agents are $N = \{1,2,3,4,5,6\}$, the items are $M=\{e_1, e_2, e_3, e_4, e_5, e_6\}$ and $w_1 = w_2 = w_3 = 1/12, w_4 = 1/6, w_5 = 1/4, w_6 = 1/3$ are the weights for agents respectively.
    The cost functions are additive with the costs of the agents for each item shown in Table~\ref{tab:exp-instance}.
    Note that the instance is IDO and we have $\WPROP_1 = \WPROP_2 = \WPROP_3 = 0.4$, $ \WPROP_4 = 0.9$, $\WPROP_5 = 1.5$, and $\WPROP_6 = 1.7$.
    
    \begin{table}[htbp]
        \centering
        \begin{tabular}{c|c|c|c|c|c|c}
            &  $e_1$ & $e_2$ & $e_3$ & $e_4$ & $e_5$ & $e_6$\\ \hline
            agent 1   & $0.7$ & $0.7$ & $0.7$ & $0.7$ & $1$ & $1$\\
            agent 2   & $0.8$ & $0.8$ & $0.8$ & $0.8$ & $0.8$ & $0.8$\\
            agent 3   & $0.7$ & $0.8$ & $0.8$ & $0.8$ & $0.8$ & $0.9$\\
            agent 4   & $0.8$ & $0.8$ & $0.8$ & $1$ & $1$ & $1$\\
            agent 5   & $1$ & $1$ & $1$ & $1$ & $1$ & $1$\\
            agent 6   & $0.8$ & $0.8$ & $0.8$ & $1$ & $1$ & $1$
        \end{tabular}
        \smallskip
        \caption{Example instance $\mathcal{I}^*$ with $6$ agents and $6$ items.}
        \label{tab:exp-instance}
    \end{table}

    After running Algorithm~\ref{alg:FBTA} on instance $\mathcal{I}^*$, we obtain the following fractional allocation (see Table~\ref{tab:fractional-allocation}), where the number in each cell corresponds to the fraction of item the agent receives.
    
    \begin{table}[htbp]
        \centering
        \begin{tabular}{c|c|c|c|c|c|c}
            &  $e_1$ & $e_2$ & $e_3$ & $e_4$ & $e_5$ & $e_6$\\ \hline
            agent 1   & $4/7$ & $0$ & $0$ & $0$ & $0$ & $0$\\
            agent 2   & $0$ & $1/2$ & $0$ & $0$ & $0$ & $0$\\
            agent 3   & $3/7$ & $1/8$ & $0$ & $0$ & $0$ & $0$\\
            agent 4   & $0$ & $3/8$ & $3/4$ & $0$ & $0$ & $0$\\
            agent 5   & $0$ & $0$ & $0$ & $1$ & $1/2$ & $0$\\
            agent 6   & $0$ & $0$ & $1/4$ & $0$ & $1/2$ & $1$
        \end{tabular}
        \smallskip
        \caption{The fractional allocation returned by the algorithm.}
        \label{tab:fractional-allocation}
    \end{table}
\end{example}

\subsection{Item-sharing Graph and Tree Structure}
In this section, we introduce a graph called \emph{item-sharing graph} that characterizes the structure of the fractional allocations.
Given a fractional allocation, we say that an item is ``shared'' by agents $i$ and $j$ if $x_i(e) > 0$ and $x_j(e) > 0$.
For any item $e$ that is shared by agents $i,j$, we can view $e$ to be an edge between $i$ and $j$, while agents are the nodes.
Naturally, such a sharing relation would induce a graph in which two agents are connected if they share an item (see Figure~\ref{fig:example-bad-graph} for an example.)
However, in such a graph, if an item is shared by many agents, a complete graph will be formed among these agents, which introduces complications to our rounding scheme.
Furthermore, we notice that even if an item $e$ is shared by multiple agents, the allocation happens in sequence.
For each agent $i\in N$, we denote by $e^i$ the last item that agent $i$ receives.
If agent $i$ reaches her (weighted) proportionality and turns inactive before she takes item $e^i$ completely, the algorithm will allocate (the remaining part of) item $e^i$ to some active agent $j \neq i$.
We call agent $j$ a \emph{successor} of agent $i$ if the algorithm tries to allocate the remaining part of item $e^i$ to $j$.
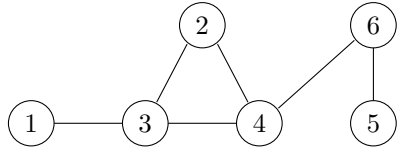
\begin{figure}[htb]
        \centering
        \begin{tikzpicture}	
            \draw (-1.5,0) circle(0.3);    \node at (-1.5,0) {$1$};
            \draw (-1.2,0)--(-0.3,0);
            
            \draw (0,0) circle(0.3);    \node at (0,0) {$3$};
            \draw (0.3,0)--(1.2,0);
            \draw (0.75,1.3) circle(0.3);    \node at (0.75,1.3) {$2$};
            \draw (0.55,1.05)--(0.15,0.28);     \draw (0.95,1.05)--(1.35,0.28);
            \draw (1.5,0) circle(0.3);    \node at (1.5,0) {$4$};
            
            \draw (1.75,0.20)--(2.75,1.1); 
            \draw (3,0) circle(0.3);    \node at (3,0) {$5$};
            \draw (3,0.3)--(3,1);
            
            \draw (3,1.3) circle(0.3);    \node at (3,1.3) {$6$};
            
        \end{tikzpicture}
        \caption{When simply viewing items as edges between agents, the graph might be complex when an item is shared by more than two agents, e.g., $e_2$ is (fractional) allocated to agents $2,3,4$.}
        \label{fig:example-bad-graph}
    \end{figure}
    
\begin{definition}[successor]
    If when agent $i$ turns inactive the item $e^{i}$ is not fully allocated and agent $j$ is the next agent who was chosen to take $e^{i}$, we call $j$ a \emph{successor} of agent $i$ with respect to $e^{i}$.
\end{definition}

\paragraph{Item-sharing Graph.}
Let the set of nodes be the agents, where each agent has a directed edge to its successor (if any).
We abuse the notation slightly and use $e^i$ to refer to the edge from agent $i$ to her successor.
We call the resulting graph $G= (N, E)$ the \emph{item-sharing graph} (see Figure~\ref{fig:example-good-graph} for an example).
Note that since each node has at most one outgoing edge, $G= (N, E)$ is a forest and $|E| \leq n-1$.

    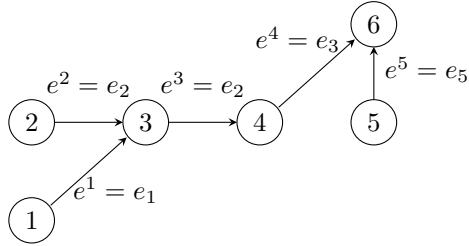
\begin{figure}[!h]
        \centering
        \begin{tikzpicture}	
            \draw (-1.5,0) circle(0.3);    \node at (-1.5,0) {$2$};
            \draw[-stealth] (-1.2,0)--(-0.3,0);
            \node at (-0.75,0.5) {$e^2 = e_2$};

            \draw (-1.5,-1.3) circle(0.3);     \node at (-1.5,-1.3) {$1$};
            \draw[-stealth] (-1.25,-1.1)--(-0.25,-0.2);
            \node at (-0.4,-0.9) {$e^1 = e_1$};
            
            \draw (0,0) circle(0.3);    \node at (0,0) {$3$};
            \draw[-stealth] (0.3,0)--(1.2,0);
            \node at (0.75,0.5) {$e^3 = e_2$};

            \draw (1.5,0) circle(0.3);    \node at (1.5,0) {$4$};
            \draw[-stealth] (1.75,0.20)--(2.75,1.1); 
            \node at (2,1.1) {$e^4 = e_3$};
            
            \draw (3,0) circle(0.3);    \node at (3,0) {$5$};
            \draw[-stealth] (3,0.3)--(3,1);
            \node at (3.7,0.7) {$e^5 = e_5$};
            
            \draw (3,1.3) circle(0.3);    \node at (3,1.3) {$6$};
            
        \end{tikzpicture}
        \caption{The item-sharing graph for the returned allocation on instance $\cI^*$.}
        \label{fig:example-good-graph}
    \end{figure}

Note that each edge in $G$ also corresponds to an item that is shared by the two endpoint agents.
Moreover, it is possible that multiple edges correspond to the same item, e.g., $e^2 = e^3 = e_2$ in the instance $\cI^*$ (see Figure~\ref{fig:example-good-graph}).
Hence the set of edges $E$ forms a (multi)set of fractional items.
By regarding the successor of each agent as its parent, we obtain a collection of rooted directed trees, where each agent with no outgoing edge is a root.
For example, the graph $G^*$ of $\cI^*$ is a tree rooted at agent $6$ (see Figure~\ref{fig:example-tree}).
We remark that an agent might have no outgoing edge for multiple reasons: 1) the agent is active when the algorithm terminates, or 2) the agent reaches her proportionality exactly when she fully consumed the item $e^i$.
%

\begin{observation}[Tree Structure]\label{obser:tree-structure}
    Given an instance $\cI = (M, N, \bw, \bc)$, let $\bx$ be the (fractional) allocation returned by the FBTA and $G = (N, E)$ be the corresponding item-sharing graph, then $G$ is a forest.
\end{observation}

    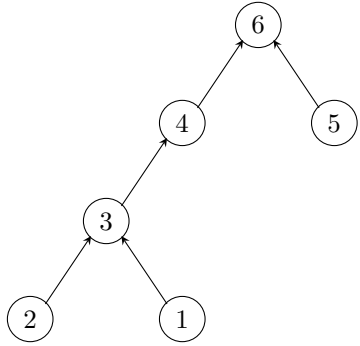
\begin{figure}[htb]
        \centering
        \begin{tikzpicture}	
            \draw (1,1.3) circle(0.3);    \node at (1,1.3) {$6$};
            \draw[-stealth] (0.2,0.2)--(0.8,1.1);     \draw[-stealth] (1.8,0.2)--(1.2,1.1);
            \draw (0,0) circle(0.3);    \node at (0,0) {$4$};
            \draw (2,0) circle(0.3);    \node at (2,0) {$5$};

            \draw[-stealth] (-0.8,-1.1)--(-0.2,-0.2);
            \draw (-1,-1.3) circle(0.3);    \node at (-1,-1.3) {$3$};
            \draw[-stealth] (-1.8,-2.4)--(-1.2,-1.5);     \draw[-stealth] (-0.2,-2.4)--(-0.8,-1.5);
            \draw (-2,-2.6) circle(0.3);    \node at (-2,-2.6) {$2$};
            \draw (0,-2.6) circle(0.3);    \node at (0,-2.6) {$1$};
        \end{tikzpicture}
        \caption{The item-sharing graph $G^*$ of instance $\cI^*$ is a tree with root of agent $6$.}
        \label{fig:example-tree}
    \end{figure}

\paragraph{Remark.}
We remark that if we construct the item-sharing graph based on a fractional allocation returned by the Moving Knife Algorithm (in the unweighted setting), then the constructed graph $G$ is a path since each agent has at most one outgoing edge and one incoming edge.
Hence, for any item-sharing graph $G$ that is a path, we can naturally generalize the results of~\cite{conf/wine/WuZZ23} to the weighted setting, which provides an (optimal) rounding scheme that bounds the total subsidy in $n/4$.

\subsection{Tree Splitting and Rounding}
In this section, we introduce the rounding scheme to turn the fractional allocation $\bx$ returned by Algorithm~\ref{alg:FBTA} to an integral allocation $\bX$.
We recall that the subsidy to each agent is automatically decided by the rounded (integral) allocation $\bX$.
We call a fractional item \emph{rounded} to agent $i$ if, in the integral allocation, the item is fully allocated to agent $i$.
We only round item $e$ to some agent $i$ with $x_i(e) > 0$ in the fractional allocation $\bx$.
Our goal is to find a rounding scheme with limited total subsidy.

Note that when the item-sharing graph $G$ is a forest, the trees in $G$ are independent to each other.
In other words, the set of fractional items of a tree is disjoint from those of any other tree.
For any two trees, we can round them independently.
Hence, in the remaining part of this paper, we assume w.l.o.g. that the item-sharing graph $G$ is a tree (instead of a forest).
We show that by splitting the tree into (edge-disjoint) subtrees and rounding each subtree independently, we can upper bound the total required subsidy to achieve proportionality.
In the splitting, we regard each subtree as a collection of edges.
We denote the size of each tree by the number of edges it contains.
By splitting based on edges, we obtain a collection of subtrees, while agents may appear in multiple trees.
Note that during the splitting, we have to ensure that edges corresponding to the same item belong to the same subtree.
As a warm-up, we consider the simple splitting that splits the graph $G$ into multiple paths, where each path contains all edges that correspond to the same fractional item (See Figure~\ref{fig:decompose-example} for an example).

    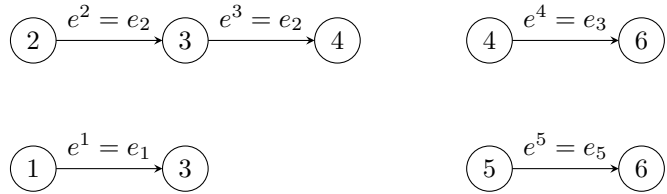
\begin{figure}[ht]
        \centering
        \begin{tikzpicture}	
            \draw (-2,0) circle(0.3);    \node at (-2,0) {$2$};
            \draw[-stealth] (-1.7,0)--(-0.3,0);
            \node at (-1,0.3) {$e^2 = e_2$};
            \draw (0,0) circle(0.3);    \node at (0,0) {$3$};
            \draw[-stealth] (0.3,0)--(1.7,0);
            \node at (1,0.3) {$e^3 = e_2$};
            \draw (2,0) circle(0.3);    \node at (2,0) {$4$};

            \draw (0,-1.7) circle(0.3);    \node at (0,-1.7) {$3$};
            \draw (-2,-1.7) circle(0.3);    \node at (-2,-1.7) {$1$};
            \draw[-stealth] (-1.7,-1.7)--(-0.3,-1.7);
            \node at (-1,-1.4) {$e^1 = e_1$};

            \draw (4,0) circle(0.3);    \node at (4,0) {$4$};
            \draw[-stealth] (4.3,0)--(5.7,0);
            \node at (5,0.3) {$e^4 = e_3$};
            \draw (6,0) circle(0.3);    \node at (6,0) {$6$};
            
            \draw (4,-1.7) circle(0.3);    \node at (4,-1.7) {$5$};
            \draw[-stealth] (4.3,-1.7)--(5.7,-1.7);
            \node at (5,-1.4) {$e^5 = e_5$};
            \draw (6,-1.7) circle(0.3);    \node at (6,-1.7) {$6$};
        \end{tikzpicture}
        \caption{A simple splitting of the item-sharing graph of instance $\cI^*$.}
        \label{fig:decompose-example}
    \end{figure}

We use the natural rounding scheme called threshold rounding~\cite{conf/wine/WuZZ23} that greedily rounds a fractional item to the agent holding the maximum fraction of the item.

\paragraph{Threshold Rounding.}
Given $\bx$ returned by Algorithm~\ref{alg:FBTA}, for every item $e\in M$, we use $N(e) = \{i\in N: x_i(e) > 0\}$ to denote the set of agents who gets (a fraction of) item $e$.
Note that if $|N(e)| = 1$ then item $e$ is integrally allocated to a single agent.
We call an item $e$ \emph{fractional} if $|N(e)| \geq 2$.
In the threshold rounding schemes, we round item $e$ to the agent holding the maximum fraction of it.
\medskip


The following theorem has been shown in~\cite{conf/wine/WuZZ23}.
As a warm-up towards our tree-splitting framework, here we give an alternative proof.

\begin{theorem}\label{theorem:n/2}
     There exists an algorithm that computes WPROP allocations with subsides at most $(n-1)/2$ for the allocation of chores to a group of weighted agents having general additive cost functions.
\end{theorem}
\begin{proof}
    We first split the tree into multiple paths where all edges of each path correspond to the same fractional item.
    For each path (each fractional item $e$), we use the threshold rounding that rounds $e$ to agent $i^* = \argmax_{i\in N} {x_i(e)}$, which incurs a subsidy
    \begin{equation*}
        \left( 1 - x_{i^*}(e) \right) \cdot c_{i^*}(e) \leq 1 - x_{i^*}(e) \leq \frac{|N(e)|-1}{|N(e)|} \leq \frac{|N(e)|-1}{2}. 
    \end{equation*}
    
    Since item $e$ corresponds to $|N(e)|-1$ edges and the total number of edges is $n-1$, by a summation of the above upper bound, we show that the total required subsidy is at most $(n-1)/2$.
\end{proof}

As we will show, if all items are shared by exactly two agents, then it is not very difficult to improve the above splitting (by ``pairing'' edges) and give a better upper bound of $n/3-1/6$ (see Section~\ref{sec:n-1}).
However, the existence of items that are shared by three or more agents complicates the splitting.
We call such items \emph{shattered items}.
Roughly speaking, every shattered item puts a constraint on how the tree can be split and may cause an increase in the required total subsidy unless we do the splitting in a very careful way.
We will discuss the case when shattered items exist in Section~\ref{sec:less-than-n-1}.

\section{Rounding Trees without Shattered Items}\label{sec:n-1}

Recall that we assume without loss of generality that the item-sharing graph $G = (N,E)$ is a tree, i.e., $|E| = n-1$.
In this section, we consider the case when there is no shattered items.
In other words, each of the edges in $E$ corresponds to a unique fractional item that is shared by its two endpoints.
We first consider the case of $n=3$, in which the tree is a path with two edges.

\subsection{Three Agents with Two Fractional Items}\label{sec:three-agents}

For the unweighted setting, Wu et al.~\cite{conf/wine/WuZZ23} showed an upper bound of $3/4$ for the total subsidy required to round paths with length $2$.
They also give an instance with three agents for which every PROPS allocation requires a total subsidy of at least $2/3$.
In this section, we close this gap by showing a rounding strategy that guarantees a WPROPS allocation with a total subsidy at most $2/3$, even for the weighted case.

\begin{theorem}\label{thm:three-agents}
    For a tree with size $2$, there exists a rounding scheme that guarantees a WPROPS allocation with a total subsidy of at most $2/3$.
\end{theorem}

For an instance $\cI = (M,N,\bw,\bc)$ where $N = \{1,2,3\}$, let $e_1$ be the item shared by agents $1,2$ and $e_2$ be the item shared by agents $2,3$ respectively.
Ignoring the directions of edges, the item-sharing graph $G$ is as Figure~\ref{fig:example-three-agents} shows.

    \begin{figure}[ht]
        \centering
        \begin{tikzpicture}	
            \draw (-2,0) circle(0.3);    \node at (-2,0) {$1$};
            \draw (-1.7,0)--(-0.3,0);   \node at(-1,0.3) {$e_1$};
            \draw (0,0) circle(0.3);    \node at (0,0) {$2$};
            \draw (0.34,0)--(1.7,0);     \node at(1,0.3) {$e_2$};
            \draw (2,0) circle(0.3);    \node at (2,0) {$3$};
        \end{tikzpicture}
        \caption{The item-sharing graph for three agents.}
        \label{fig:example-three-agents}
    \end{figure}
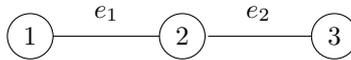

We first introduce four different rounding schemes, that we call $\LL$, $\RR$, $\LR$, and $\RL$.
We use $\LL$ to denote the scheme that rounds $e_1$ to agent $1$ and $e_2$ to agent $2$; $\RR$ to denote the scheme that rounds $e_1$ to agent $2$ and $e_2$ to agent $3$; $\LR$ to denote the scheme that rounds $e_1$ to agent $1$ and $e_2$ to agent $3$;
$\RL$ to denote the scheme that rounds both $e_1$ and $e_2$ to agent $2$.
In the following, we upper bound the required subsidies for the four rounding schemes.
We use $s_{\LL}, s_{\RR}, s_{\LR}, s_{\RL}$ to denote the subsidy required in the four rounding strategies respectively.
Let $a_1 = c_2(e_1), a_2 = c_2(e_2)$, we assume w.l.o.g. that $a_1 \geq a_2$ and let $\alpha = a_2/a_1\in (0,1]$.
For any variable $x$, we use $(x)^+$ to denote $\max \{x,0\}$.

\begin{lemma}\label{lemma:four-upperbounds}
    For the four different rounding schemes, the total required subsidy can be upper bounded as shown in the following table.
    \begin{table}[htbp]
        \centering
        \begin{tabular}{|c|c|}
        \hline
        \textbf{Rounding Scheme} & \textbf{Upper Bound for Total Subsidy}   \\ \hline
        $\LL$ & $x_2(e_1) + \left((1-x_2(e_2)) \cdot \alpha - x_2(e_1)\right)^+$ \\ \hline
        $\RR$ & $x_2(e_2) + \left((1-x_2(e_1)) - x_2(e_2) \cdot \alpha \right)^+$ \\ \hline
        $\LR$ & $x_2(e_1) + x_2(e_2)$ \\ \hline
        $\RL$ & $(1-x_2(e_1)) + (1-x_2(e_2)) \cdot \alpha$          \\ \hline
        \end{tabular}
        \caption{Upper bounds for the subsidies required for the four different rounding schemes.}\label{tab:three-upperbounds}
    \end{table}
\end{lemma}

Given the above characterized upper bounds, we are ready to show Theorem~\ref{thm:three-agents}.

\begin{proposition}\label{proposition:xy<1/4}
    For any $0< x, y <1$, we have $\min \{xy, (1-x)(1-y)\} \leq 1/4$
\end{proposition}
\begin{proof}
    When $x+y \geq 1$, we have $(1-x)(1-y)\leq x\cdot(1-x)$; when $x+y < 1$, we have $xy\leq x\cdot (1-x)$.
    In both cases we have $\min \{xy, (1-x)(1-y)\} \leq x\cdot (1-x) \leq 1/4$.
\end{proof}

\begin{proofof}{Theorem~\ref{thm:three-agents}}
    For notational simplicity, we use $y_1$ and $y_2$ to denote $x_2(e_1)$ and $x_2(e_2)$.
    Recall from Lemma~\ref{lemma:four-upperbounds} that we have the following bounds:
    \begin{table}[htbp]
        \centering
        \begin{tabular}{|c|c|}
        \hline
        \textbf{Rounding Scheme} & \textbf{Upper Bound for Total Subsidy}   \\ \hline
        $\LL$ & $y_1 + \left((1-y_2) \cdot \alpha - y_1\right)^+$ \\ \hline
        $\RR$ & $y_2 + \left((1-y_1) - y_2 \cdot \alpha \right)^+$ \\ \hline
        $\LR$ & $y_1 + y_2$ \\ \hline
        $\RL$ & $(1-y_1) + (1-y_2) \cdot \alpha$          \\ \hline
        \end{tabular}
    \end{table}
    
    To prove Theorem~\ref{thm:three-agents}, we discuss different cases depending on whether $((1-y_2)\cdot \alpha - y_1)^+ = 0$ and $((1-y_1)- y_2\cdot \alpha )^+ = 0$.
Note that $(1-y_2)\cdot \alpha - y_1 \leq (1-y_1)- y_2\cdot \alpha$.
Hence there are three different cases to discuss.

\paragraph{Case 1.}
When $((1-y_2)\cdot \alpha - y_1)^+ = 0$ and $((1-y_1) - y_2\cdot \alpha)^+ = 0$, we have $s_\LL \leq y_1$ and $s_\RR \leq y_2$.
By summing up $s_\LL, s_\RR, s_\RL$, we have
\begin{align*}
    s_\LL + s_\RR + s_\RL &\leq y_1 + y_2 + (1-y_1) + (1-y_2) \cdot \alpha \\
    &\leq y_1 + y_2 + (1-y_1) + (1-y_2) = 2,
\end{align*}
which leads to $\min \{s_\LL, s_\RR, s_\RL\} \leq 2/3.$

\paragraph{Case 2.}
When $((1-y_2)\cdot \alpha - y_1)^+ = 0$ and $((1-y_1) - y_2\cdot \alpha)^+ > 0$, we have $s_\LL \leq y_1$ and $s_\RR \leq (1-y_1) + y_2 \cdot (1-\alpha)$.
By Proposition~\ref{proposition:xy<1/4}, we have $\min \{ y_2 \cdot (1-\alpha), (1-y_2) \cdot \alpha\} \leq 1/4$.
Hence we have
\begin{equation*}
    \min \{s_\RR, s_\RL\} \leq (1-y_1) + \frac{1}{4} = \frac{5}{4} - y_1.
\end{equation*}
By summing up $s_\LL$ and $\min \{s_\RR, s_\RL\}$, we have
\begin{equation*}
    s_\LL + \min \{s_\RR, s_\RL\} \leq \frac{5}{4} \quad \Rightarrow \quad \min \{s_\LL, s_\RR, s_\RL\} \leq \frac{5}{8} < \frac{2}{3}.
\end{equation*}




\paragraph{Case 3.}
Finally, we consider the case that $((1-y_2)\cdot \alpha - y_1)^+ > 0$ and $((1-y_1) - y_2\cdot \alpha)^+ > 0$.
In this case we have $s_{\LL} \leq (1-y_2)\cdot \alpha$ and $s_{\RR}\leq (1-y_1)+y_2\cdot(1-\alpha)$.
Then we have
\begin{equation*}
    \frac{2-\alpha}{\alpha} \cdot s_\LL + s_\RR + s_\LR = (2-\alpha) \cdot (1-y_2) + (1+(2-\alpha) \cdot y_2) = 3-\alpha.
\end{equation*}
Hence we have
\begin{equation*}
    \min \{s_\LL, s_\RR, s_\LR\} \leq \frac{\alpha\cdot(3-\alpha)}{2+\alpha} \leq \frac{2}{3}.
\end{equation*}

In conclusion, for all cases, there exists a rounding scheme that guarantees WPROPS allocation with a total subsidy of at most $2/3$.
\end{proofof}

Given Theorem~\ref{thm:three-agents}, we argue that for any instance, there exists a rounding strategy that guarantees a WPROPS allocation with a total subsidy of at most $n/3-1/6$.
Our algorithm is based on a tree splitting that splits the item-sharing graph $G$ into subtrees with size at most two.
For each subtree, we do rounding independently and bound the subsidy by $2/3$.

\subsection{Simple Splitting}

Recall that during the splitting, we regard each tree as a collection of edges.
If we can regard two adjacent edges as a subtree, we can bound the subsidy for the tree by $2/3$, following Theorem~\ref{thm:three-agents}.
Since there are $n-1$ edges, by repeatedly splitting subtrees with size two, we can bound the total subsidy by $\lceil (n-1)/2 \rceil \cdot 2/3 \leq n/3$.
In the following, we introduce the \splitting, which we use to split any item-sharing graph to a collection of subtrees with size two.

\begin{algorithm}[htbp]
    \caption{\splitting}
    \label{alg:SD}
    \KwIn{An item-sharing graph $G$.}
    Let $i$ be the deepest (furthest from the root) node in $G$ \;
    Let $j$ be $i$'s parent \;
    \If{$j$ has another child $l\neq i$}{
        Take edges $e^i, e^l$ as one subtree \;
    }
    \ElseIf{$j$ is the root}{
        \Return \ ; \qquad\qquad\qquad \tcp{$G$ has only one edge}
    }
    \Else{
        Take edges $e^i, e^j$ as one subtree \;
    }
    \KwOut{A subtree $T$ with a size of two and the remaining graph $G'$.}
\end{algorithm}

\begin{lemma}\label{lemma:simple-splitting}
    After repeatedly running \splitting on graph $G$, we obtain a collection of trees with size two, and at most one of them has size one.
\end{lemma}
\begin{proof}
    We prove the lemma by mathematical induction on the size of $G$.
    When $G$ has size at most two, \splitting directly outputs $G$ as the subtree.
    Next we consider $n\geq 3$.
    Assume the induction hypothesis holds that for any tree with size at most $n-1$.
    It suffices to show that in \splitting, the remaining graph (with size $n-2$) after trimming two edges is still connected.
    Recall that in \splitting, we start from the deepest leaf $i$ and trace back to its parent $j$.
    When $j$ has another child $l$, we trim two edges $e^i, e^l$.
    Since $i$ is the deepest leaf, $l$ is also a leaf node and thus the remaining graph is still connected.
    When $i$ is the only child of $j$, we trim $e^i$ and $e^j$, in which case the remaining graph is also connected.
    In either case, the remaining graph (with size $n-2$) is connected, and we can use the induction hypothesis to split it into a collection of trees with size two, and at most one of them has size one.
\end{proof}

\begin{example}
Consider an instance $\cI'$ that has the same item-sharing graph as Figure~\ref{fig:example-tree}, but the items $e^2, e^3$ are different.
We use it as an example to show the execution of \splitting.
The splitting result is shown in Figure~\ref{fig:decompose-nochild}.

    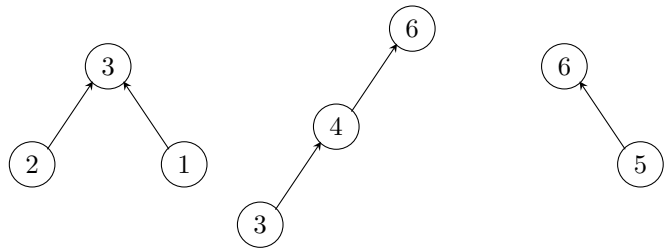
\begin{figure}[ht]
        \centering
        \begin{tikzpicture}	
            \draw (3,0.8) circle(0.3);    \node at (3,0.8) {$6$};
            \draw[-stealth] (3.8,-0.3)--(3.2,0.6);
            \draw (4,-0.5) circle(0.3);    \node at (4,-0.5) {$5$};

            \draw (1,1.3) circle(0.3);    \node at (1,1.3) {$6$};
            \draw[-stealth] (0.2,0.2)--(0.8,1.1); 
            \draw (0,0) circle(0.3);    \node at (0,0) {$4$};
            \draw[-stealth] (-0.8,-1.1)--(-0.2,-0.2);
            \draw (-1,-1.3) circle(0.3);    \node at (-1,-1.3) {$3$};

            \draw (-3,0.8) circle(0.3);    \node at (-3,0.8) {$3$};
            \draw[-stealth] (-3.8,-0.3)--(-3.2,0.6);     
            \draw[-stealth] (-2.2,-0.3)--(-2.8,0.6);
            \draw (-4,-0.5) circle(0.3);    \node at (-4,-0.5) {$2$};
            \draw (-2,-0.5) circle(0.3);    \node at (-2,-0.5) {$1$};
        \end{tikzpicture}
        \caption{Illustration for the execution of \splitting on $\cI'$.}
        \label{fig:decompose-nochild}
    \end{figure}
\end{example}

Given Theorem~\ref{thm:three-agents}, we can round each subtree independently and bound the subsidy required for each tree by $2/3$.
We show that the total subsidy can be upper bounded by $n/3$.

\begin{theorem}\label{thm:subsidy-n-1}
    Given the fractional allocation $\bx$ with $n-1$ fractional items returned by Algorithm~\ref{alg:FBTA}, there exists a rounding scheme that returns an integral WPROPS allocation with total subsidy at most $(n-1)/3$ when $n$ is odd; at most $n/3 - 1/6$ when $n$ is even.
\end{theorem}
\begin{proof}
    Considering the item-sharing graph $G=(N,E)$ with $|E| = n-1$.
    Recall that in each splitting, we get a subtree with size two.
    When $n$ is odd, after \splitting, we obtain exactly $(n-1)/2$ trees with size two.
    By Theorem~\ref{thm:three-agents}, for each tree, we can round the fractional items with a total subsidy of at most $2/3$.
    Hence the total subsidy across all subtrees is at most 
    \begin{equation*}
        \frac{2}{3} \cdot \frac{n-1}{2} = \frac{n-1}{3}.
    \end{equation*}

    When $n$ is even, after \splitting, we obtain $(n-2)/2$ trees with size two and one tree with size one.
    For the tree with size one, we use threshold rounding which incurs a subsidy of at most $1/2$.
    Hence the total subsidy across all subtrees is at most 
    \begin{equation*}
        \frac{2}{3} \cdot \frac{n-2}{2} + \frac{1}{2} = \frac{n-2}{3} + \frac{1}{2} = \frac{n}{3} - \frac{1}{6}.\qedhere
    \end{equation*}
\end{proof}

\section{Rounding Trees with Shattered Item}\label{sec:less-than-n-1}

In this section, we consider case when there are at least one shattered item.
Note that the \splitting proposed in Section~\ref{sec:n-1} might not work in this case when multiple edges correspond to the same fractional item.
Consider the allocation for instance $\cI^*$, returned by Algorithm~\ref{alg:FBTA}.
The item-sharing graph $G$ of the allocation is as Figure~\ref{fig:atom-path-example} shows, where the two red edges $e^2$ and $e^3$ correspond to the same item $e_2$.
Following the \splitting edges $e^2$ and $e^1$ would be taken as a subtree, while $e^3$ and $e^4$ are put into another subtree.
Recall that for each subtree, we apply the rounding scheme as in Theorem~\ref{thm:three-agents} independently.
However, we notice that items $e^2$ and $e^3$ can not be rounded independently.
We define the structure like $e^2, e^3$ as an \emph{atom-path}, which contains edges that cannot be separated.

\begin{definition}[Atom-path]
    A path $i \to \cdots \to j$ is called an \emph{atom-path} if all edges in the path correspond to the same item, and the path is maximal.
\end{definition}

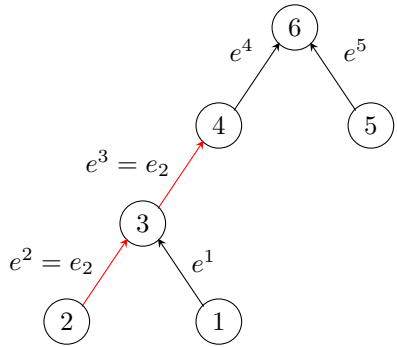
\begin{figure}[ht]
    \centering
    \begin{tikzpicture}	
        \draw (1,1.3) circle(0.3);    \node at (1,1.3) {$6$};
        \draw[-stealth] (0.2,0.2)--(0.8,1.1);     \draw[-stealth] (1.8,0.2)--(1.2,1.1);
        \draw (0,0) circle(0.3);    \node at (0,0) {$4$};
        \draw (2,0) circle(0.3);    \node at (2,0) {$5$};
        \node at (0.3,1) {$e^4$};     \node at (1.8,1) {$e^5$}; 
        
        \draw[red][-stealth] (-0.8,-1.1)--(-0.2,-0.2);
        \draw (-1,-1.3) circle(0.3);    \node at (-1,-1.3) {$3$};
        \draw[red][-stealth] (-1.8,-2.4)--(-1.2,-1.5);     \draw[-stealth] (-0.2,-2.4)--(-0.8,-1.5);
        \draw (-2,-2.6) circle(0.3);    \node at (-2,-2.6) {$2$};
        \draw (0,-2.6) circle(0.3);    \node at (0,-2.6) {$1$};
        \node at (-2.2,-1.8) {$e^2 = e_2$};     \node at (-1.2,-0.5) {$e^3=e_2$}; 
        \node at (-0.2,-1.8) {$e^1$};
    \end{tikzpicture}
    \caption{The item-sharing graph $G^*$ of instance $\cI^*$.}
    \label{fig:atom-path-example}
\end{figure}

We note that any atom-path has size at least $2$.
When there exists an atom-path in $G$, we need to ensure that the edges in the atom-path are not separated into different subtrees.
In the following, we introduce the splitting and bound the total subsidy by $n/3$ as follows.

\begin{lemma}\label{lemma:general-tree-subsidy}
    For any tree with size $z \geq 2$ such that there exists an atom-path in the tree, there exists a rounding scheme with total subsidy bounded by $z/3$.
\end{lemma}

We remark that Lemma~\ref{lemma:general-tree-subsidy} directly leads to the following theorem.

\begin{theorem}\label{thm:subsidy-<n-1}
    Given the fractional allocation $\bx$ with less than $n-1$ fractional items returned by Algorithm~\ref{alg:FBTA}, there exists a rounding scheme that returns an integral WPROPS allocation with total subsidy at most $(n-1)/3$.
\end{theorem}

\subsection{Atom-path splitting}

In this section, we introduce the general tree splitting that we use when there exists at least one atom-path in the item-sharing graph $G$.
A naive idea is to take the atom-path alone (with size at least two) as a subtree and recursively do the splitting on the remaining graph (which might have multiple connected components).
For example, in Figure~\ref{fig:atom-path-example}, we can split the tree into the atom-path $\{e^2, e^3\}$, a subtree with edges $\{e^4, e^5\}$, and another subtree with edge $e^1$.
Then we can round the subtrees independently as in Section~\ref{sec:n-1}.
However, notice that rounding a path with size one, e.g., item $e^1$ in Figure~\ref{fig:atom-path-example}, may cause an increase in subsidy by $1/2$.
This may cause a problem when the length of the atom-path is large, as each of these ``hanging edges'' will incur an increase of $1/2$ in the subsidy.
It would be even harder to bound the total subsidy if there are other atom-paths in some connected components of the resulting graph after removing the atom-path.
However, we notice that there is an easy fix to the above problem: we can simply attach the ``hanging edges'' to the atom-path to form a subtree, which we call an \emph{expanded atom-path} (the formal definition will be given later).
On a high level, we want to do the splitting such that after taking out an expanded atom-path, each connected component of the resulting graph is \emph{good}.

\begin{definition}[Good Subtree]
    We say that a subtree is \emph{good} if either it contains at least one atom-path, or it has even size.
\end{definition}

Note that we can bound (using mathematical induction) the total subsidy required to round a good subtree with size $z$ by $z/3$ using Lemma~\ref{lemma:general-tree-subsidy} (if it contains an atom-path) or Theorem~\ref{thm:subsidy-n-1} (if it has even size).
To ensure that we can split the tree into an expanded atom-path and a collection of good subtrees, we check each connected component after taking out the atom-path:
\begin{itemize}
    \item If the component is a good subtree, we keep it directly and bound the required subsidy by $1/3$ of its size, following Lemma~\ref{lemma:general-tree-subsidy} or Theorem~\ref{thm:subsidy-n-1}.
    \item For any connected component with odd size and without an atom-path, we show that there exists an edge in the connected component that can be attached to the atom-path such that after removing the edge the remaining graph is a collection of good subtrees.
\end{itemize}

We conclude the splitting in Algorithm~\ref{alg:GD} and call it \textsf{Atom-path Splitting}.

\begin{algorithm}[htbp]
    \caption{\textsf{Atom-path Splitting}}
    \label{alg:GD}
    \KwIn{An item-sharing graph $G$ with atom-path.}
    Pick an arbitrary atom-path and let it be a subtree \;
    Remove the edges of atom-path from the original tree \;
    \For{each connected component $T$}{
        \If{$T$ has an atom-path or has even size}{
            Let $T$ be one subtree of the splitting \;
        }
        \Else(\textit{$T$ has odd size and no atom-path}){
            Let $(i,j)\in T$ be an edge such that (1) agent $i$ is in the atom-path and (2) the subtree containing $j$ after removing $(i,j)$ has even size\;
            Attached $(i,j)$ to the atom-path\;
            Do \splitting on each connected component of $T\setminus\{(i,j)\}$ \;
        }
    }
    \KwOut{A collection of subtrees and the remaining graph $G'$.}
\end{algorithm}

By running \textsf{Atom-path Splitting} for a graph $G$ with atom-path(s), we obtain a collection of good subtrees and an atom-path with multiple edges attached to it (see Example~\ref{example:atom-path-splitting}).

\begin{example} \label{example:atom-path-splitting}
Consider a tree graph $G$ with an atom-path $(e^1, e^2, e^3)$.
After cutting the atom-path we obtain four subtrees $T_1, T_2, T_3, T_4$ (see Figure~\ref{fig:general-splitting-example}).
By definition $T_2$ and $T_4$ are good subtrees.

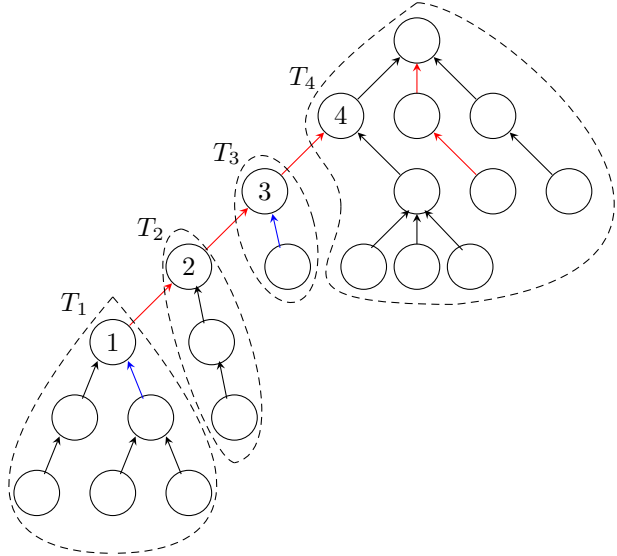
\begin{figure}[ht]
    \centering
    \begin{tikzpicture}	
        \draw (0,0) circle (0.3); \node at (0,0) {$1$};
        \draw[red][-stealth] (0.22,0.22)--(0.78,0.78);
        \draw (1,1) circle (0.3); \node at (1,1) {$2$};
        \draw[red][-stealth] (1.22,1.22)--(1.78,1.78);
        \draw (2,2) circle (0.3); \node at (2,2) {$3$};
        \draw[red][-stealth] (2.22,2.22)--(2.78,2.78);
        \draw (3,3) circle (0.3); \node at (3,3) {$4$};
        %

        \node at (-0.5,0.5) {$T_1$};
        \draw[densely dashed] (0,0.6) .. controls (2.2,-2) and (1.4,-2.8) .. (0,-2.8);
        \draw[densely dashed] (0,0.6) .. controls (-2.2,-2) and (-1.4,-2.8) .. (-0,-2.8);
        \draw (-0.5,-1) circle (0.3);
        \draw[-stealth] (-0.4,-0.75)--(-0.2,-0.25);
        \draw (-1,-2) circle (0.3);
        \draw[-stealth] (-0.9,-1.75)--(-0.7,-1.25);
        \draw (0.5,-1) circle (0.3);
        \draw[blue][-stealth] (0.4,-0.75)--(0.2,-0.25);
        \draw (0,-2) circle (0.3);
        \draw[-stealth] (0.1,-1.75)--(0.3,-1.25);
        \draw (1,-2) circle (0.3);
        \draw[-stealth] (0.9,-1.75)--(0.7,-1.25);
        
        \node at (0.5,1.5) {$T_2$};
        \draw[densely dashed] (0.9,1.5) .. controls (1.5,1.4) and (2.5,-1.4) .. (1.6,-1.6);
        \draw[densely dashed] (0.9,1.5) .. controls (0.2,1.4) and (1.2,-1.5) .. (1.6,-1.6);
        \draw (1.3,0) circle (0.3);
        \draw[-stealth] (1.2,0.25)--(1.1,0.75); 
        \draw (1.6,-1) circle (0.3);
        \draw[-stealth] (1.5,-0.75)--(1.4,-0.25); 

        \node at (1.5,2.5) {$T_3$};
        \draw[densely dashed] [rotate around ={15:(2.15,1.5)}] (2.15, 1.5) ellipse (0.5 and 1);
        \draw (2.3,1) circle (0.3);
        \draw[blue][-stealth] (2.2,1.25)--(2.1,1.7);

        \node at (2.5,3.5) {$T_4$};
        \draw[densely dashed] (4,4.5) .. controls (5,4.5) and (9,1) .. (4.5,0.5);
        \draw[densely dashed] (4,4.5) .. controls (1.45,2.65) and (3,2.7) .. (3,1.8);
        \draw[densely dashed] (3,1.8) .. controls (3,1) and (2,0.5) .. (4.5,0.5);
        \draw (4,2) circle (0.3);
        \draw[-stealth] (3.78,2.22)--(3.22,2.78); 
        \draw (3.3,1) circle (0.3);
        \draw[-stealth] (3.4,1.25)--(3.9,1.75);    
        \draw (4,1) circle (0.3);  
        \draw[-stealth] (4,1.3)--(4,1.7); 
        \draw (4.7,1) circle (0.3);
        \draw[-stealth] (4.6,1.22)--(4.1,1.75); 
        \draw[-stealth] (3.22,3.22)--(3.78,3.78);
        \draw (4,4) circle (0.3); 
        \draw[red][-stealth] (4,3.3)--(4,3.7);
        \draw (4,3) circle (0.3);    
        \draw[red][-stealth] (4.78,2.22)--(4.22,2.78);
        \draw (5,2) circle (0.3);
        \draw[-stealth] (4.78,3.22)--(4.22,3.78);
        \draw (5,3) circle (0.3);       
        \draw[-stealth] (5.78,2.22)--(5.22,2.78);
        \draw (6,2) circle (0.3);
    \end{tikzpicture}
    \caption{An example tree with an atom-path $(e^1, e^2, e^3)$.}
    \label{fig:general-splitting-example}
\end{figure}

Since $T_3$ has size one, we can simply attached the edge to the atom-path.
Observe that $T_1$ also has odd size.
Hence there must exits an edge (the edge with blue color in Figure~\ref{fig:general-splitting-example}) that is incident to node $1$ whose removal results in a collection of even-size subtrees.
The splitting result is as follows (see Figure~\ref{fig:general-splitting-result}).

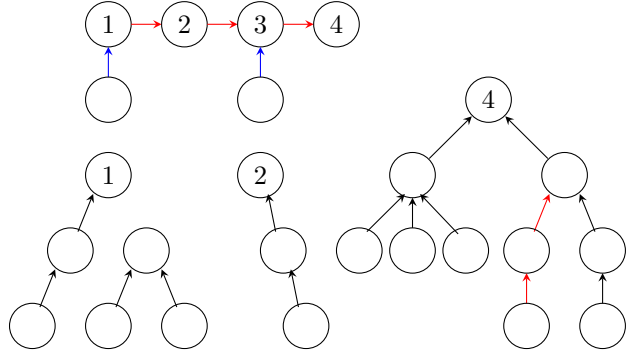
\begin{figure}[ht]
    \centering
    \begin{tikzpicture}	
        \draw (0,4) circle (0.3); \node at (0,4) {$1$};
        \draw[red][-stealth] (0.3,4)--(0.7,4);
        \draw (1,4) circle (0.3); \node at (1,4) {$2$};
        \draw[red][-stealth] (1.3,4)--(1.7,4);
        \draw (2,4) circle (0.3); \node at (2,4) {$3$};
        \draw[red][-stealth] (2.3,4)--(2.7,4);
        \draw (3,4) circle (0.3); \node at (3,4) {$4$};
        \draw[blue][-stealth] (2,3.3)--(2,3.7);
        \draw (2,3) circle (0.3);
        \draw[blue][-stealth] (0,3.3)--(0,3.7);
        \draw (0,3) circle (0.3);
        %

        \draw (0,2) circle (0.3); \node at (0,2) {$1$};
        \draw (-0.5,1) circle (0.3);
        \draw[-stealth] (-0.4,1.25)--(-0.2,1.75);
        \draw (-1,0) circle (0.3);
        \draw[-stealth] (-0.9,0.25)--(-0.7,0.75);

        \draw (0.5,1) circle (0.3);
        \draw (0,0) circle (0.3);
        \draw[-stealth] (0.1,0.25)--(0.3,0.75);
        \draw (1,0) circle (0.3);
        \draw[-stealth] (0.9,0.25)--(0.7,0.75);
        
        \draw (2,2) circle (0.3);   \node at (2,2) {$2$}; 
        \draw (2.3,1) circle (0.3);
        \draw[-stealth] (2.2,1.25)--(2.1,1.75); 
        \draw (2.6,0) circle (0.3);
        \draw[-stealth] (2.5,0.25)--(2.4,0.75); 


        \draw (4,2) circle (0.3);
        \draw[-stealth] (4.22,2.22)--(4.78,2.78);
        \draw (3.3,1) circle (0.3);
        \draw[-stealth] (3.4,1.25)--(3.9,1.75);    
        \draw (4,1) circle (0.3);  
        \draw[-stealth] (4,1.3)--(4,1.7); 
        \draw (4.7,1) circle (0.3);
        \draw[-stealth] (4.6,1.22)--(4.1,1.75); 
        \draw (5,3) circle (0.3);      \node at (5,3) {$4$};       
        \draw[-stealth] (5.78,2.22)--(5.22,2.78);
        \draw (6,2) circle (0.3);
        \draw[red][-stealth] (5.6,1.25)--(5.8,1.75);
        \draw (5.5,1) circle (0.3);
        \draw[-stealth] (6.4,1.25)--(6.2,1.75);
        \draw (6.5,1) circle (0.3);
        \draw[red][-stealth] (5.5,0.3)--(5.5,0.7);
        \draw (5.5,0) circle (0.3);
        \draw[-stealth] (6.5,0.3)--(6.5,0.7);
        \draw (6.5,0) circle (0.3);
    \end{tikzpicture}
    \caption{The splitting result of Figure~\ref{fig:general-splitting-example}.}
    \label{fig:general-splitting-result}
\end{figure}
\end{example}

Note that each node in the atom-path is attached at most one edge.
We call the resulting structure an \emph{expanded atom-path}.

\begin{definition}[Expanded Atom-path]
    We call a subtree an \emph{expanded atom-path} if it is composed of an atom-path and a collection of attached edges, where each node in the atom-path is incident to at most one attached edge.
\end{definition}

\begin{lemma}\label{lemma:atom-path-splitting}
    For any graph $G$ containing (at least) one atom-path, Algorithm~\ref{alg:GD} splits $G$ into an expanded atom-path and a collection of good subtrees.
\end{lemma}
\begin{proof}
    To prove the lemma, it suffices to consider the else condition in line $6$ as in other cases the subtree $T$ is a good tree.
    Let $T$ be a subtree with odd size and no atom-path.
    Let $i$ be the node in both the atom-path and $T$.
    Let $(i,j_1), \ldots, (i,j_t)$ be the edges incident to $i$ in $T$, and let $T'_x$ be the subtree containing $j_x$ after the hypothetical deletion of edge $(i,j_x)$.
    Since $T$ has odd size, at least one of $T'_1,\ldots,T'_t$ must have even size.
    Suppose $T'_x$ has even size.
    We attach edge $(i,j_x)$ to the atom-path and delete it from $T$.
    Then we obtain two subtrees: $T'_x$ and $T\setminus (T'_x \cup {(i,j_x)})$, both of which have even size.
    Since we attach at most one edge of each node of the atom-path, the result of the splitting is an expanded atom-path and a collection of good subtrees.
\end{proof}

Recall that we can upper bound the subsidy required to round a good subtree with $z$ edges by $z/3$.
To upper bound the total subsidy of a graph with at least one atom-path, it remains to bound the subsidy required to round an expanded atom-path.

\begin{lemma}\label{lemma:expanded-atom-path}
    For any expanded atom-path consisting of an atom-path with $k$ edges and $h\leq k+1$ attached edges, there exists a rounding scheme with total subsidy bounded by $(k+h)/3$.
\end{lemma}

We defer the proof of Lemma~\ref{lemma:expanded-atom-path} to Section~\ref{ssec:expanded-atom-path}.
Given Lemma~\ref{lemma:expanded-atom-path}, we are ready to prove Lemma~\ref{lemma:general-tree-subsidy}.


\begin{replemma}{lemma:general-tree-subsidy}
    For any tree with size $z \geq 2$ such that there exists an atom-path in the tree, there exists a rounding scheme with total subsidy bounded by $z/3$.
\end{replemma}

\begin{proofof}{Lemma~\ref{lemma:general-tree-subsidy}}
    We prove the lemma by induction on $z$.
    We first consider the base case of a tree with $z=2$ edges, which is an atom-path.
    In such case, we only have one fractional item that is shared by three agents.
    Using threshold rounding, we round the item to the agent holding the maximum fraction of it, which incurs a subsidy of at most $2/3$.
    Hence the lemma holds for the base case.
    
    Now suppose that for any tree with edges $z' < z$ containing at least one atom-path, there exists a rounding scheme with total subsidy at most $z'/3$.
    We consider a tree with $z$ edges containing at least one atom-path.
    Following Lemma~\ref{lemma:atom-path-splitting}, given any tree with at least one atom-path, the \textsf{Atom-path splitting} returns an expanded atom-path and a collection of subtrees that either contains an atom-path or has even size.
    For the former case we can bound its required subsidy by $z'/3$, where $z'$ is the number of edges it contains, by induction hypothesis.
    For the later case we can also bound the required subsidy by $z'/3$ by \splitting and Theorem~\ref{thm:subsidy-n-1}.
    Finally, by Lemma~\ref{lemma:expanded-atom-path}, for expanded atom-path we can also bound the required subsidy by $1/3$ of its size.
    Since the summation of the size of all good subtrees and the expanded atom-path is $z$, the total subsidy required is at most $z/3$, and the lemma follows.
\end{proofof}

\subsection{Expanded Atom-path}\label{ssec:expanded-atom-path}

In this section, we consider the expanded atom-paths and prove Lemma~\ref{lemma:expanded-atom-path}.
We first define some notation.
Fix any expanded atom-path $T$, we assume that the atom-path has $k\geq 2$ edges, which correspond to the one single item denoted by $e_0$.
We denote the nodes in the atom-path by agents $1, \ldots, k+1$, where agent $i+1$ is the successor of agent $i$, for all $i=1,2,\ldots,k$.
We denote the edge attached to agent $i$ by $e_i$.
When there is no attached edge to agent $i$, we define $e_i = \bot$.
We call agents $1,2,\ldots,k+1$ the \emph{atom-path agents} and the other agents the \emph{attached agents}.
We use $h \leq k+1$ to denote the number of attached edges/agents in the expanded atom-path.
For each item $e_i$, where $i\in \{1,\ldots,k+1\}$, we use $y_i$ to denote the fraction of item $e_i$ held by the corresponding attached agent, i.e., $y_i = 1 - x_i(e_i)$.
For all $i\in \{1, \ldots, k+1\}$, when there is no ambiguity, we use $x_i$ to denote $x_i(e_0)$ for convenience.

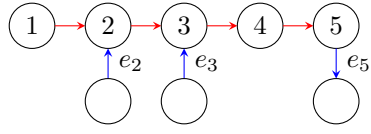
\begin{figure}[ht]
    \centering
    \begin{tikzpicture}	
        \draw (0,4) circle (0.3); \node at (0,4) {$1$};
        \draw[red][-stealth] (0.3,4)--(0.7,4);
        \draw (1,4) circle (0.3); \node at (1,4) {$2$};
        \draw[red][-stealth] (1.3,4)--(1.7,4);
        \draw (2,4) circle (0.3); \node at (2,4) {$3$};
        \draw[red][-stealth] (2.3,4)--(2.7,4);
        \draw (3,4) circle (0.3); \node at (3,4) {$4$};
        \draw[red][-stealth] (3.3,4)--(3.7,4);
        \draw (4,4) circle (0.3); \node at (4,4) {$5$};
        \draw[blue][-stealth] (4,3.7)--(4,3.3);
        \node at (4.3,3.5) {$e_5$};
        \draw (4,3) circle (0.3);
        \draw[blue][-stealth] (2,3.3)--(2,3.7);
        \node at (2.3,3.5) {$e_3$};
        \draw (2,3) circle (0.3);
        \draw[blue][-stealth] (1,3.3)--(1,3.7);
        \node at (1.3,3.5) {$e_2$};
        \draw (1,3) circle (0.3);
    \end{tikzpicture}
    \caption{Example of an expanded atom-path with $k=4$ and $h=3$, where the red edges correspond to item $e_0$, the three blue edges correspond to items $e_2, e_3$ and $e_5$ respectively.}
\end{figure}

Now we are ready to upper bound the subsidy required for expanded atom-paths.
We prove Lemma~\ref{lemma:expanded-atom-path} by giving the following arguments for different types of expanded atom-paths respectively.

\begin{lemma}\label{lemma:expanded-atom-path-general}
    For any expanded atom-path with $h \leq k \cdot (2-\frac{6}{k+1})$, there exists a rounding scheme with total subsidy at most $(k+h)/3$.
\end{lemma}

\begin{lemma}\label{lemma:expanded-atom-path-k+1}
    For any expanded atom-path with $h = k+1$, there exists a rounding scheme with total subsidy at most $(k+h)/3$.
\end{lemma}

\begin{lemma}\label{lemma:expanded-atom-path-k}
    For any expanded atom-path with $h = k$, there exists a rounding scheme with total subsidy at most $(k+h)/3$.
\end{lemma}

\begin{lemma}\label{lemma:expanded-atom-path-k-1}
    For any expanded atom-path with $2 \leq k \leq 3$ and $h = k-1$, there exists a rounding scheme with total subsidy at most $(k+h)/3$.
\end{lemma}

We first show that given the above arguments, we can prove Lemma~\ref{lemma:expanded-atom-path}.
\begin{replemma}{lemma:expanded-atom-path}
    For any expanded atom-path consisting of an atom-path with $k$ edges and $h\leq k+1$ attached edges, there exists a rounding scheme with total subsidy bounded by $(k+h)/3$.
\end{replemma}

\begin{proofof}{Lemma~\ref{lemma:expanded-atom-path}}
    Recall that $h \leq k+1$.
    Following Lemma~\ref{lemma:expanded-atom-path-k+1} and Lemma~\ref{lemma:expanded-atom-path-k}, for all the atom-paths with $h \in \{k+1,k\}$, there exists a rounding scheme with total subsidy bounded by $(k+h)/3$.
    It can be verified that $k - 1 \leq k \cdot (2 - \frac{6}{k+1})$ holds for all $k\geq 4$.
    Hence for all $k\geq 4$ and $h\leq k-1$ we can bound the total subsidy by $(k+h)/3$ using Lemma~\ref{lemma:expanded-atom-path-general}.
    Furthermore, it can be verified that $k - 1 \leq k \cdot (2 - \frac{6}{k+1})$ holds when $k = 3, h\leq 1$, and when $k=2, h=0$.
    Hence it remains to consider the cases that $k =3, h=2$; or $k =2, h=1$, which are covered by Lemma~\ref{lemma:expanded-atom-path-k-1}.
    In conclusion, for all types of expanded atom-path, there exists a rounding scheme with total subsidy at most $(k+h)/3$.
\end{proofof}

It remains to prove Lemma~\ref{lemma:expanded-atom-path-general},~\ref{lemma:expanded-atom-path-k+1},~\ref{lemma:expanded-atom-path-k} and ~\ref{lemma:expanded-atom-path-k-1}.
We use a charging argument that charges money to the fractional items $e_0, e_1, \ldots, e_{k+1}$: we charge each fractional item $e_i$ an amount of money $p_i$.
In the following analysis, we focus on bounding the money we charge for each item $e_i$ and we show that the total charge to the fractional items is sufficient to pay for the incurred subsidy.

\begin{lemma:expanded-atom-path-general}
    For an expanded atom-path with $h \leq k \cdot (2-\frac{6}{k+1})$, there exists a rounding scheme with total subsidy at most $(k+h)/3$.
\end{lemma:expanded-atom-path-general}

\begin{proof}
    We consider the threshold rounding that rounds each item to the agent that holds the maximum fraction of it.
    For item $e_0$, we round it to $i^* = \argmax_{1\leq i \leq k+1} x_i$.
    For each item $e_i$, we round it to agent $i$ if and only if $y_i \leq 1/2$.
    We charge item $e_0$ an amount of money $p_0 = \frac{k}{k+1}$.
    For each item $e_i$ we charge it an amount of money $p_i = 1/2$ if and only if item $e_i \neq \bot$.
    For completeness, we let $p_i = 0$ when $e_i = \bot$.

    Since the inclusion of item $e_0$ to agent $i^*$ incurs an increase in subsidy by at most
    \begin{equation}
        (1-x_{i^*}) \cdot c_{i^*}(e_0) \leq 1-x_{i^*} \leq \frac{k}{k+1} = p_0,
    \end{equation}
    and the inclusion of item $e_i$ for $i\neq 0$ incurs an increase in subsidy by at most $1/2$, the charged money is sufficient to pay for the subsidy.
    To bound the required subsidy by $(k+h)/3$, it remains to prove that
    \begin{equation*}
        \sum_{0 \leq i \leq k+1} p_i = \frac{k}{k+1} + \frac{h}{2} \leq \frac{k+h}{3},
    \end{equation*}
    which is equivalent to $h \leq k \cdot (2-\frac{6}{k+1})$ as stated in the lemma.
\end{proof}

Notice that in the above analysis, rounding item $e_0$ to agent $i^*$ releases a space of $x_i \cdot c_i(e_0)$ for all agent $i \neq i^*$, which should be helpful in reducing the subsidy requirement for the attached edges.
In the following, we make use of this observation to prove Lemma~\ref{lemma:expanded-atom-path-k+1},~\ref{lemma:expanded-atom-path-k} and ~\ref{lemma:expanded-atom-path-k-1}.

\begin{observation}\label{obser:c_i(e_0)>=c_i(e_i)}
    For any agent $i \leq k$, $c_i(e_0) \geq c_i(e_i)$.
\end{observation}
\begin{proof}
    We remark that for any agent $i \leq k$, $e_0$ is the last (fraction of) item that she receives.
    Otherwise, agent $i$ would not have a successor with respect to $e_0$.
    In other words, for any agent $i$, the event that $i$ takes (a fraction of) item $e_i$ happens before the event that $i$ takes (a fraction of) item $e_0$.
    Recall that we consider IDO instances and allocate items with increasing cost in Algorithm~\ref{alg:FBTA},
    Hence, for any agent $i \leq k$, we have $c_i(e_0) \geq c_i(e_i)$.
\end{proof}

Note that the above observation might not hold for agent $k+1$: it is possible that $e^{k+1} = e_{k+1} \neq e_0$.
In this case we have $c_{k+1}(e_0) \leq c_{k+1}(e_{k+1})$.
Given the above observation, for all agent $i\neq i^*$ that does not receive item $e_0$, when rounding item $e_i$ we should somehow give agent $i$ a slightly higher priority (compared with the other endpoint) as rounding $e_i$ to agent $i$ incurs a subsidy $(y_i\cdot c_i(e_i) - x_i\cdot c_i(e_0))^+ \leq (y_i - x_i)^+$.
We formalize this idea into the following rounding scheme.


\paragraph{Biased Threshold Rounding.}
Given an attached edge $e_i$, if $y_i - x_i \leq 1-y_i$, we round item $e_i$ to agent $i$; otherwise we round $e_i$ to the attached agent.

\medskip


\begin{proposition}\label{proposition:modified-greedy}
    For any agent $i\in \{1,2,\ldots,k+1\}$ that does not receive item $e_0$ and satisfies $c_i(e_0) \geq c_i(e_i)$, the subsidy incurred by rounding $e_i$ using the biased threshold rounding is at most $(1-x_i)/2$.
\end{proposition}
\begin{proof}
    By the biased threshold rounding, if $y_i - x_i \leq 1-y_i$ then we round item $e_i$ to agent $i$, which incurs a subsidy of $(y_i\cdot c_i(e_i) - x_i \cdot c_i(e_0))^+ \leq (y_i - x_i)^+$.
    Otherwise, we round item $e_i$ to the other agent which incurs an increase in subsidy by at most $(1-y_i)$.
    Hence the incurred subsidy can be bounded by
    \begin{equation*}
        \min \{1-y_i, (y_i- x_i)^+\}
        \leq \frac{1-y_i+y_i-x_i}{2} = \frac{1-x_i}{2}.\qedhere
    \end{equation*}
\end{proof}


\begin{lemma:expanded-atom-path-k+1}
    For an expanded atom-path with $h = k+1$, there exists a rounding scheme with total subsidy at most $(k+h)/3$.
\end{lemma:expanded-atom-path-k+1}

\begin{proof}
    Recall from Observation~\ref{obser:c_i(e_0)>=c_i(e_i)} that we have $c_{i}(e_0) \geq c_{i}(e_{i})$ for all $i\in {1,2,\ldots,k}$.
    Unfortunately, we can not guarantee this for agent $k+1$.
    In the following, we present different rounding schemes depending on whether $c_{k+1}(e_0) \geq c_{k+1}(e_{k+1})$. 
    We first consider the case when $c_{k+1}(e_0) < c_{k+1}(e_{k+1})$.

    In this case, we round $e_0$ to agent $k+1$ and round items $e_1, \ldots, e_k$ using the biased threshold rounding.
    By Proposition~\ref{proposition:modified-greedy}, rounding each $e_i$ incurs a subsidy of at most $(1-x_i)/2$, for all $i\in {1,2,\ldots,k}$.
    Next, we show that at least one of the following rounding schemes gives a total subsidy of at most $(k+h)/3$:
    \begin{itemize}
        \item \textbf{Scheme 1.} We round item $e_{k+1}$ to the attached agent;
        \item \textbf{Scheme 2.} We round item $e_{k+1}$ to agent $k+1$.
    \end{itemize}
    
    Note that rounding item $e_{k+1}$ to the attached agent release a space of $(1-y_{k+1}) \cdot c_{k+1}(e_{k+1})$ for agent $k+1$.
    Hence the incurred subsidy by rounding $e_0$ in scheme 1 is
    \begin{align*}
        &((1-x_{k+1}) \cdot c_{k+1}(e_0) - (1-y_{k+1}) \cdot c_{k+1}(e_{k+1}))^+    \\
        \leq & ((y_{k+1}-x_{k+1}) \cdot c_{k+1}(e_{k+1}))^+ &\tag{$c_{k+1}(e_0) < c_{k+1}(e_{k+1})$}  \\
        \leq & (y_{k+1}-x_{k+1})^+. &\tag{$c_{k+1}(e_{k+1}) \leq 1$}
    \end{align*}
    
    Since we need to compensate at most $1-y_{k+1}$ to the attached agent for receiving $e_{k+1}$, the total subsidy required in scheme 1 can be upper bounded by 
    \begin{equation}
        \sum_{1 \leq i \leq k} \frac{1-x_i}{2} + (y_{k+1} - x_{k+1})^+ +  (1-y_{k+1}). \label{eq:k+1-scheme1}
    \end{equation}
    
    Similarly, for scheme 2, the total required subsidy can be upper bounded by
    \begin{equation}
        \sum_{1 \leq i \leq k} \frac{1-x_i}{2} + (1 - x_{k+1}) + y_{k+1}. \label{eq:k+1-scheme2}
    \end{equation}
    
    When $y_{k+1} - x_{k+1} \geq 0$, it is easy to find that Equation~\eqref{eq:k+1-scheme1} dominates Equation~\eqref{eq:k+1-scheme2} (gives a smaller upper bound).
    We can upper bound the subsidy for scheme 1 as follows.
    \begin{align*}
        \|s\|_1 &\leq (1 - x_{k+1}) + \sum_{1 \leq i \leq k} \frac{1-x_i}{2}  &\tag{By Equation~\eqref{eq:k+1-scheme1} and $y_{k+1} - x_{k+1} \geq 0$}\\
        & = \sum_{1 \leq i \leq k} x_i + \frac{k}{2} - \sum_{1 \leq i \leq k} \frac{x_i}{2}   &\tag{$1 - x_{k+1} = \sum_{1 \leq i \leq k} x_i$}\\
        & = \frac{k}{2} + \sum_{1 \leq i \leq k} \frac{x_i}{2} \leq \frac{k+1}{2}.  &\tag{$\sum_{1 \leq i \leq k} x_i \leq 1$}
    \end{align*}
    
    Next we consider the case that $y_{k+1} - x_{k+1} < 0$.
    We have
    \begin{align*}
        \text{For scheme 1, } \|s\|_1 &\leq \sum_{1 \leq i \leq k} \frac{1-x_i}{2} + (1-y_{k+1});  &\tag{By Equation~\eqref{eq:k+1-scheme1}}\\
        \text{For scheme 2, } \|s\|_1 &\leq \sum_{1 \leq i \leq k} \frac{1-x_i}{2} + (1 - x_{k+1}) + y_{k+1}.   &\tag{By Equation~\eqref{eq:k+1-scheme2}}
    \end{align*}
    
    Summing up the above two upper bounds, we have 
    \begin{align*}
        &\sum_{1 \leq i \leq k} \frac{1-x_i}{2} + (1-y_{k+1}) + (1 - x_{k+1}) + \sum_{1 \leq i \leq k} \frac{1-x_i}{2} + y_{k+1} \\
        = &\sum_{1 \leq i \leq {k+1}} (1-x_i) + 1  
        = k+2 - \sum_{1 \leq i \leq {k+1}} x_i  = k+1  &\tag{$\sum_{1 \leq i \leq {k+1}} x_i = 1$}
    \end{align*}
    
    Therefore, at least one of the two rounding schemes requires a total subsidy of at most
    \begin{equation*}
        \frac{k+1}{2} \leq \frac{2k+1}{3} = \frac{k+h}{3},
    \end{equation*}
    where the inequality holds for all $k\geq 2$.

    \smallskip
    
    Next, we consider the case when $c_{k+1}(e_0) \geq c_{k+1}(e_{k+1})$.

    Following Observation~\ref{obser:c_i(e_0)>=c_i(e_i)}, in this case we have $c_i(e_0) \geq c_i(e_i)$ for all $i \leq \{1,2,\ldots,k+1\}$ and can assume w.l.o.g. that $x_1\leq x_2\leq \cdots\leq x_{k+1}$.
    We first round item $e_0$ to agent $k+1$, who holds the maximum fraction of $e_0$.
    Let $\alpha = c_{k+1}(e_{k+1})$.
    As before, by using biased threshold rounding to round items $e_1,\ldots, e_k$, we can bound the required subsidy by $\sum_{i=1}^k \frac{1-x_i}{2}$.
    Again, we consider the following two rounding schemes:
    \begin{itemize}
        \item \textbf{Scheme 1.} We round item $e_{k+1}$ to the attached agent;
        \item \textbf{Scheme 2.} We round item $e_{k+1}$ to agent $k+1$.
    \end{itemize}
    
    Note that rounding item $e_{k+1}$ to the attached agent release a space of $(1-y_{k+1}) \cdot \alpha$ for agent $k+1$.
    Hence the subsidy required in scheme 1 can be upper bounded by
    \begin{equation}
        \sum_{1 \leq i \leq k} \frac{1-x_i}{2} + ((1 - x_{k+1}) - (1-y_{k+1}) \cdot \alpha)^+ + (1-y_{k+1}). \label{eq:k+1-scheme3}
    \end{equation}
    
    For scheme 2, the total required subsidy can be upper bounded by
    \begin{equation}
        \sum_{1 \leq i \leq k} \frac{1-x_i}{2} + (1 - x_{k+1}) + y_{k+1}\cdot \alpha. \label{eq:k+1-scheme4}
    \end{equation}
    
    When $(1 - x_{k+1}) - (1-y_{k+1}) \cdot \alpha \geq 0$, the subsidy required in scheme 1 can be upper bounded by
    \begin{equation}
        \sum_{1 \leq i \leq k} \frac{1-x_i}{2} + (1 - x_{k+1})  + (1-y_{k+1})\cdot (1-\alpha) \label{eq:k+1-scheme3-1}
    \end{equation}

    By Proposition~\ref{proposition:xy<1/4}, we have
    \begin{equation*}
        \min \{(y_{k+1}\cdot \alpha), (1-y_{k+1})\cdot (1-\alpha) \} \leq 1/4. \tag{$0\leq y_{k+1}\leq 1$ and $0 \leq \alpha \leq 1$}
    \end{equation*}
    
    Hence the minimum of~\eqref{eq:k+1-scheme4} and~\eqref{eq:k+1-scheme3-1} can be bounded by (recall that $x_{k+1} = \max_{1 \leq i\leq k+1} \{x_i\}$)
    \begin{align*}
        &\sum_{1 \leq i \leq k} \frac{1-x_i}{2} +  (1 - x_{k+1})  + \frac{1}{4} 
        = \sum_{1 \leq i \leq k} \frac{1-x_i}{2} + \sum_{1 \leq i \leq k} x_i + \frac{1}{4}   \\
        = & \frac{k}{2} + \frac{1}{2} \cdot \sum_{1 \leq i \leq k} x_i + \frac{1}{4} 
        \leq \frac{k}{2} + \frac{k}{2k+2} + \frac{1}{4} = \frac{2k^2 + 5k + 1}{4k+4} \leq \frac{2k+1}{3} = \frac{k+h}{3},
    \end{align*}
    where the last inequality holds for all $k\geq 2$.

    Now suppose that $(1-x_{k+1}) - (1-y_{k+1}) \cdot \alpha <0$.
    We have the total required subsidy for scheme 1 can be bounded by
    \begin{equation}
        \sum_{1 \leq i \leq k} \frac{1-x_i}{2} + (1-y_{k+1}).  \label{eq:k+1-scheme3-2}
    \end{equation}
    
    Summing up the upper bounds for scheme 1 (Equation~\eqref{eq:k+1-scheme3-2}) and scheme 2 (Equation~\eqref{eq:k+1-scheme4}), the total subsidy required by the two rounding schemes combined is at most
    \begin{align*}
        &\sum_{1 \leq i \leq k} \frac{1-x_i}{2} + (1-y_{k+1}) + \sum_{1 \leq i \leq k} \frac{1-x_i}{2} + (1 - x_{k+1}) + y_{k+1} \cdot \alpha \\
        = &\sum_{1 \leq i \leq {k+1}} (1-x_i) + 1 + y_{k+1}\cdot \alpha - y_{k+1}   \\
        \leq & \sum_{1 \leq i \leq {k+1}} (1-x_i) + 1
        = k+2 - \sum_{1 \leq i \leq {k+1}} x_i = k+1.
    \end{align*}
    
    Therefore, at least one of the two rounding schemes requires a total subsidy at most $(k+1)/2$, which is not larger than $(k+h)/3$ for all $k\geq 2$.
\end{proof}

\begin{lemma:expanded-atom-path-k}
    For an expanded atom-path with $h = k$, there exists a rounding scheme with total subsidy at most $(k+h)/3$.
\end{lemma:expanded-atom-path-k}

\begin{proof}
    When $h = k$, there exists an agent $j$ in the atom-path with no edge attached, i.e., $e_j = \bot$.
    We first consider the case that $j \neq k+1$ and $c_{k+1}(e_0) < c_{k+1}(e_{k+1})$.
    \begin{itemize}
        \item \textbf{Scheme 1.}
        We round item $e_0$ to agent $k+1$ and round item $e_{k+1}$ to the attached agent.
        We use the biased threshold rounding for other items $e_i$ s.t. $1\leq i\leq k$.
        %
        By a similar argument as in Lemma~\ref{lemma:expanded-atom-path-k+1}, the total subsidy required can be upper bounded by
        \begin{equation}
            \sum_{1 \leq i \leq k} \frac{1-x_i}{2} - \frac{1-x_j}{2} + (y_{k+1} - x_{k+1})^+ + (1-y_{k+1}). \label{eq:k-scheme1}
        \end{equation}
        \item  \textbf{Scheme 2.}
        We round item $e_0$ to agent $j$ and item $e_{k+1}$ by threshold rounding.
        For any other item $e_i$ s.t. $1 \leq i\leq k$, we use the biased threshold rounding.
        The total subsidy can be upper bounded by
        \begin{equation}
            \sum_{1 \leq i \leq k} \frac{1-x_i}{2} - \frac{1-x_j}{2} + (1-x_j) + \frac{1}{2}. \label{eq:k-scheme2}
        \end{equation}
    \end{itemize}
    
    Summing up the above two upper bounds, we have
    \begin{align*}
        &  \sum_{1 \leq i \leq k} (1-x_i) + (y_{k+1} - x_{k+1})^+ + (1-y_{k+1}) + \frac{1}{2} \\
        = &\frac{k}{2} - \sum_{1 \leq i \leq k} x_i +  (y_{k+1} - x_{k+1})^+ + (1-y_{k+1}) + \frac{1}{2} \\
        = & k + x_{k+1} +  (y_{k+1} - x_{k+1})^+ - y_{k+1} + \frac{1}{2}    \\
        = & k + \frac{1}{2}. &\tag{When $y_{k+1} - x_{k+1} \geq 0$}
    \end{align*}
    
    Hence when $y_{k+1} - x_{k+1} \geq 0$, at least one of the two rounding schemes requires a total subsidy at most $(2k+1)/4$, which is no larger than $\frac{2k}{3} = \frac{k+h}{3}$ for all $k\geq 2$.

    When $y_{k+1} - x_{k+1} < 0$, Equation~\eqref{eq:k-scheme1} becomes
    \begin{equation}
        \sum_{1 \leq i \leq k} \frac{1-x_i}{2} - \frac{1-x_j}{2} + (1-y_{k+1}) \label{eq:k-scheme1a}
    \end{equation}

    We introduce one more rounding scheme.
    \begin{itemize}
        \item \textbf{Scheme 3.}
        We round items $e_0$ to $e_j$ and $e_{k+1}$ to agent $k+1$.
        For any other item $e_i$ s.t. $1 \leq i\leq k$, we use the biased threshold rounding.
        The total subsidy can be bounded by
        \begin{equation}
            \sum_{1 \leq i \leq k} \frac{1-x_i}{2} - \frac{1-x_j}{2} + (1-x_j) + y_{k+1} - x_{k+1}\cdot \alpha \label{eq:k-scheme3}
        \end{equation}
    \end{itemize}
    
    By summing up Equations~\eqref{eq:k-scheme1a} and~\eqref{eq:k-scheme3}, we have
    \begin{align*}
        \sum_{1 \leq i \leq k} (1-x_i) + (1-x_{k+1} \cdot \alpha)
        =  k + (1-\sum_{1 \leq i \leq k} x_i) - x_{k+1} \cdot \alpha 
        = k + x_{k+1} - x_{k+1} \cdot \alpha.
    \end{align*}
    
    To show that one of the two rounding schemes requires a total subsidy of at most $\frac{k}{2} + \frac{1}{4} \leq \frac{2k}{3}$, it suffices to show that $x_{k+1} \cdot (1-\alpha) \leq \frac{1}{2}$.
    Assume otherwise $x_{k+1} \cdot (1-\alpha) > \frac{1}{2}$.
    Following Proposition~\ref{proposition:xy<1/4}, we have $(1-x_{k+1}) \cdot \alpha < \frac{1}{4}$.
    Consider the rounding scheme that rounds item $e_0$ to agent $k+1$ and uses threshold rounding to any other items.
    It can be verified that the total subsidy required is at most
    \begin{equation*}
        (1-x_{k+1}) \cdot \alpha + \frac{1}{2} \cdot k \leq \frac{k}{2} + \frac{1}{4} \leq \frac{2k}{3}.
    \end{equation*}
    
    Hence, for the case that $j \neq k+1$ and $c_{k+1}(e_0) < c_{k+1}(e_{k+1})$, there always exists a rounding scheme with total subsidy bounded by $(k+h)/3$.

    \smallskip

    Next we consider the case that $j = k+1$ or $c_{k+1}(e_{k+1}) \geq c_{k+1}(e_0)$, in which we have $c_i(e_i) \leq c_i(e_0)$ for all $i\neq j$.
    We round item $e_0$ to agent $j$ and use the biased threshold rounding to round every item $e_i$ s.t. $i\neq j$.
    Following similar analysis as shown before, total subsidy required is at most
    \begin{equation*}
        \sum_{i\in [k+1], i\neq j} \frac{1-x_i}{2} + (1-x_j) = \frac{k}{2} + \frac{1 - x_j}{2}.
    \end{equation*}
    
    To ensure that $\frac{k}{2} + \frac{1 - x_j}{2} $, we need $x_j \geq 1 - \frac{k}{3}$, which naturally holds for all $k\geq 3$.

    It remains to consider the case that $k = 2$ and $x_j < 1 - \frac{2}{3} = \frac{1}{3}$.
    We assume w.l.o.g. that $x_{3} = \max\{x_1,x_2,x_3\}$ and we round $e_0$ to agent $3$.
    Note that $j \neq 3$ since $x_3 \geq 1/3$.
    Let $\alpha = c_3(e_3)$.
    We consider two rounding schemes that both round item $e_0$ to agent $3$ but differ in the rounding of item $e_3$.
    Let $i$ be the agent other than $j$ and $3$, and we use the biased threshold rounding to round item $e_i$.
    Following Proposition~\ref{proposition:modified-greedy}, the total subsidy required can be upper bounded by
        \begin{align*}
            \begin{cases}
                \cfrac{1-x_i}{2} + ((1-x_3) - (1-y_3)\alpha)^+ + (1-y_3), &\text{if we round $e_3$ to the attached agent} \\
                \cfrac{1-x_i}{2} + (1-x_3) + y_3\cdot \alpha, &\text{if we round $e_3$ to agent $3$}.
            \end{cases}
        \end{align*}
        
    When $(1-x_3) - (1-y_3)\alpha \geq 0$, following Proposition~\ref{proposition:xy<1/4}, the minimum of two bounds is at most
    \begin{align*}
        \frac{1-x_i}{2} + (1-x_3) + \frac{1}{4} 
        = \frac{5}{4} + \frac{1-x_3-x_i}{2} - \frac{x_3}{2}
        = \frac{5}{4} + \frac{x_j}{2} - \frac{x_3}{2}
        \leq \frac{5}{4} \leq \frac{4}{3}.   
    \end{align*}
        
    Hence at least one of the two rounding schemes with total subsidy bounded by $4/3= 2k/3$.
        
    When $(1-x_3) - (1-y_3)\alpha < 0$, the sum of the above two bounds is
    \begin{align*}
        (1-x_i) + (1-x_3) + (1-y_3) + y_3\cdot \alpha
        \leq 3 - x_1 - x_3
        \leq 2 + x_j \leq \frac{7}{3}. 
    \end{align*}
    
    Hence at least one of the two rounding schemes requires subsidy at most $7/6 < 4/3= 2k/3$.
\end{proof}

\begin{replemma}{lemma:expanded-atom-path-k-1}
    For any expanded atom-path with $2 \leq k \leq 3$ and $h = k-1$, there exists a rounding scheme with total subsidy at most $(k+h)/3$.
\end{replemma}

\begin{proof}
    For $k = 3$, the upper bound is $\frac{2k-1}{3} = \frac{5}{3}$.
    Note that there exists two agents $i,j$ such that ${e_i} = {e_j} = \bot$.
    We first consider the case that $k+1 = 4 \in \{i,j\}$.
    Let $p,q$ be the two agents other than $i,j$.
    Consider the rounding scheme that rounds $e_0$ to $\argmax \{x_i,x_j\}$, say agent $i$.
    Applying the biased threshold rounding to items $e_p$ and $e_q$, we can bound the total subsidy required by 
    \begin{equation*}
        \frac{1-x_p}{2} + \frac{1-x_q}{2} + (x_p + x_q + x_j) = 1 + \frac{x_p + x_q + 2\cdot x_j}{2} \leq 1 + \frac{x_p + x_q + x_j+ x_i}{2} = \frac{3}{2}.
    \end{equation*}
        
    Now suppose that $k+1 = 4 \notin \{i,j\}$.
    Let $p$ be the agent other than $i,j,4$.
    Again, we round $e_0$ to $\argmax \{x_i,x_j\}$, say agent $i$, and apply the biased threshold rounding to round item $e_p$.
    By applying the threshold rounding for item $e_4$, we can upper bound the total subsidy by
    \begin{equation*}
        \frac{1-x_p}{2} +1-x_i + \frac{1}{2} \leq \frac{3}{2} + \frac{1 - x_p - x_i - x_j}{2} 
        = \frac{3+x_4}{2}.
    \end{equation*}
    
    Note that $\frac{3+x_4}{2} \leq \frac{5}{3}$ holds when $x_4 \leq 1/3$.
    Assume otherwise that $x_4 > 1/3$.
    We round $e_0$ to agent $4$ and use the biased threshold rounding for item $e_p$.
    %
    The total subsidy can be upper bounded by
    \begin{equation*}
        \begin{cases}
            \frac{1-x_p}{2} + ((1-x_4)-(1-y_4)\cdot\alpha))^+ + (1-y_4),\quad &\text{if we round $e_4$ to agent $4$} \\
            \frac{1-x_p}{2} + (1-x_4) + y_4\cdot \alpha, &\text{if we round $e_4$ to the attached agent}.
        \end{cases}
    \end{equation*}
    
    When $(1-x_4)-(1-y_4)\cdot\alpha \geq 0$, by Proposition~\ref{proposition:xy<1/4}, the minimum of two upper bounds is at most
    \begin{equation*}
        \frac{1-x_p}{2} + (1-x_4) + \frac{1}{4} \leq \frac{3}{4} + (1-x_4) \leq \frac{3}{4} + \frac{1}{3} \leq \frac{5}{3}.
    \end{equation*}
    
    When $(1-x_4)-(1-y_4)\cdot\alpha < 0$, the first upper bound becomes
    \begin{equation*}
        \frac{1-x_p}{2} + (1-y_4) \leq \frac{3}{2}\leq \frac{5}{3}.
    \end{equation*}
    
    For $k=2$, the upper bound is $\frac{2k-1}{3} = 1$.
    Note that there exist two agents $i,j$ such that ${e_i} = {e_j} = \bot$.
    Let $p$ be the agent other than $i,j$ (who has an attached edge).
    We first consider the case that $p \neq k+1$.
    We round $e_0$ to $\argmax \{x_i,x_j\}$, say agent $i$, and apply the biased threshold rounding for items $e_p$.
    The total subsidy can be upper bounded by
    \begin{equation*}
        \frac{1-x_p}{2} + (x_j+x_p) = \frac{1+x_p + 2\cdot x_j}{2} \leq \frac{1+x_p + x_j+ x_i}{2} = 1.
    \end{equation*}
        
    Consider otherwise that $p = k+1$, where the atom-path reduces to a line.
    Following Theorem 4.10 in~\cite{conf/wine/WuZZ23}, there exists a rounding scheme that requires a total subsidy of at most $n/4 = 1$.
\end{proof}

\section{Conclusion and Open Problems}
\label{sec:conclusion}

In this work we revisit the problem of computing weighted proportional allocations with subsidy, and propose a rounding scheme for the fractional bid-and-take algorithm based on tree splitting.
We show that for any instance with $n$ agents with general weights, we can compute a weighted proportional allocation that requires a total subsidy of at most $n/3$.
Our work leaves several natural problems open.
The most fascinating open problem is to investigate the existence of algorithms that guarantee weighted proportionality requiring a total subsidy of at most $n/4$.
We conjecture that the fractional bid-and-take algorithm admits such rounding schemes.
It would also be interesting to study other fairness notions like Maximin Share or AnyPrice Share and see if we can further lower the subsidy requirement.
Finally, we believe that it would be interesting to apply the tree composition-based framework to other fair allocation settings, e.g., the Best of Both Worlds setting, that require rounding fractional allocations.

\bibliographystyle{abbrv}
\bibliography{subsidy}

\newpage
\appendix

\section{Missing Proofs}\label{sec:missing-proofs}
\begin{replemma}{lemma:reduction-to-IDO}
    If there exists a polynomial time algorithm that given any IDO instance computes a WPROPS allocation with total subsidy at most $\alpha$, then there exists a polynomial time algorithm that given any instance computes a WPROPS allocation with total subsidy at most $\alpha$.
\end{replemma}
\begin{proof}
    For any general instance $\mathcal{I} = (M,N,\bw,\bc)$, we reduce it to an IDO instance $\mathcal{I'} = (M,N,\bw,\bc')$ by the following construction.
    Let $\sigma_i(k)\in M$ be the $k$-th most costly item under cost function $c_i$.
    Let $c'_i(e_k) = c_i(\sigma_i(k))$.
    Thus with cost function $\bc'$, the instances $\mathcal{I'}$ is IDO.
    Note that for all $i\in N$ we have $c'_i(M) = c_i(M)$.
    Then we run the algorithm for the IDO instance $\mathcal{I'}$ and get a WPROPS allocation $\bX'$ with subsidy $\bs'$ such that $\|\bs'\|_1 \leq \alpha$.
    By definition, for all agent $i\in N$ we have
    \begin{equation*}
        c'_i(X'_i) - s'_i \leq w_i \cdot c'_i(M) = w_i \cdot c_i(M) = \WPROP_i.
    \end{equation*}
    
    We use $\bX'$ to guide us on computing a WPROPS allocation $\bX$ with subsidy $\bs$ for the original instance $\mathcal{I}$.
    We show that $c_i(X_i) \leq c'_i(X'_i)$ for all $i\in N$, which implies $\|\bs\|_1 \leq \|\bs'\|_1 \leq \alpha$.
    We initialize $X_i = \emptyset$ for all $i\in N$ and let $P = M$ be the set of unallocated items.
    Sequentially for $j = m,\ldots,1$, we let the agent $i$ who receives item $e_j$ under allocation $\bX'$, i.e., $e_j \in X'_i$, pick her favorite unallocated item.
    At the beginning of each round $j$, we have $|P| = j$.
    Since $e_j$ is the $j$-th most costly item under cost function $c'_i$, we must have $c_i(e) \leq c'_i(e_j)$ for the item $e$ agent $i$ picks during round $j$.
    Therefore for all items in $X_i$ and $X'_i$ we can establish a one-to-one correspondence relationship similarly, which implies $c_i(X_i) \leq c'_i(X'_i)$.
\end{proof}

\begin{replemma}{lemma: wprop}
    The output allocation $\bx$ is a fractional WPROP allocation. 
\end{replemma}
\begin{proof}
    Note that throughout the allocation process, we maintain the property that $c_i(\bx_i) \leq \WPROP_i$ for each agent $i$.
    Particularly, for every inactive agent $i\in N \setminus A$, we have $c_i(\bx_i) = \WPROP_i = w_i \cdot c_i(M)$.
    Hence it suffices to show all items are fully allocated, i.e., there is at least one active agent when we try to allocate each item $e$.
    Since we allocate each item $e$ to the active agent $i$ with the minimum $\frac{c_i(e)}{c_i(M)}$, we have $\frac{c_i(e)}{c_i(M)} \leq \frac{c_j(e)}{c_j(M)}$ for every active $j\neq i$. Therefore we have $\frac{c_j(X_i)}{c_j(M)} \geq \frac{c_i(X_i)}{c_i(M)}$ for any active agents $i,j\in A$.

    Assume by contradiction that when assigning some item $e$, all agents are inactive, i.e., $A=\emptyset$. Let $j$ be the last agent that becomes inactive.
    Consider the moment when $j$ becomes inactive, we have
    \begin{equation*}
        1 = \frac{c_j(M)}{c_j(M)}>\sum_{i\in N} \frac{c_j(\bx_i)}{c_j(M)}
        \ge \sum_{i\in N}\frac{c_i(\bx_i)}{c_i(M)} = \sum_{i\in N} w_i = 1,
    \end{equation*}
    which is a contradiction.
\end{proof}

\begin{replemma}{lemma:four-upperbounds}
    For the four different rounding schemes, the total required subsidy can be upper bounded as shown in the following table.
    \begin{table}[htbp]
        \centering
        \begin{tabular}{|c|c|}
        \hline
        \textbf{Rounding Scheme} & \textbf{Upper Bound for Total Subsidy}   \\ \hline
        $\LL$ & $x_2(e_1) + \left((1-x_2(e_2)) \cdot \alpha - x_2(e_1)\right)^+$ \\ \hline
        $\RR$ & $x_2(e_2) + \left((1-x_2(e_1)) - x_2(e_2) \cdot \alpha \right)^+$ \\ \hline
        $\LR$ & $x_2(e_1) + x_2(e_2)$ \\ \hline
        $\RL$ & $(1-x_2(e_1)) + (1-x_2(e_2)) \cdot \alpha$          \\ \hline
        \end{tabular}
        \caption{Upper bounds for the subsidies required for the four different rounding schemes.}\label{tab:four-upperbounds}
    \end{table}
\end{replemma}
\begin{proof}
    In the $\LL$ rounding scheme, we round both items $e_1, e_2$ to their left endpoint agents.
    The inclusion of item $e_1$ to agent $1$ incurs an increase in subsidy by at most
    \begin{equation*}
        (1-x_1(e_1))\cdot c_1(e_1) \leq 1-x_1(e_1) = x_2(e_1).
    \end{equation*}
    
    Note that rounding $e_1$ to agent $1$ releases a space of $x_2(e_1) \cdot a_1$ for agent $2$.
    Hence the inclusion of item $e_2$ to agent $2$ incurs an increase in subsidy by at most
    \begin{equation*}
        \left ((1-x_2(e_2)) \cdot a_2 - x_2(e_1) \cdot a_1 \right)^+ \leq a_1 \cdot \left ((1-x_2(e_2)) \cdot \alpha - x_2(e_1) \right)^+ \leq \left ((1-x_2(e_2)) \cdot \alpha - x_2(e_1) \right)^+.
    \end{equation*}
    Hence we have
    \begin{equation*}
        s_\LL \leq x_2(e_1) + \left((1-x_2(e_2)) \cdot \alpha - x_2(e_1)\right)^+.
    \end{equation*}
    
    In the $\RR$ rounding scheme, we round both items $e_1, e_2$ to their right endpoint agents.
    The inclusion of item $e_2$ to agent $3$ incurs an increase in subsidy by at most
    \begin{equation*}
        (1-x_3(e_2))\cdot c_3(e_2) \leq 1-x_3(e_2) = x_2(e_2).
    \end{equation*}
    
    Since rounding $e_2$ to agent $3$ releases a space of $x_2(e_2) \cdot a_2$ for agent $2$, the inclusion of item $e_2$ to agent $2$ incurs an increase in subsidy by at most
    \begin{equation*}
        \left ((1-x_2(e_1)) \cdot a_1 - x_2(e_2) \cdot a_2 \right)^+ \leq a_1 \cdot \left ((1-x_2(e_1)) - x_2(e_2) \cdot \alpha \right)^+ \leq \left ((1-x_2(e_1)) - x_2(e_2) \cdot \alpha \right)^+.
    \end{equation*}
    Hence we have
    \begin{equation*}
        s_\LL \leq x_2(e_2) + \left((1-x_2(e_1)) - x_2(e_2) \cdot \alpha \right)^+.
    \end{equation*}

    In the $\LR$ rounding scheme, we round items $e_1$ to agent $1$ and $e_2$ to agent $3$.
    The total subsidy incurred can be upper bounded as follows:
    \begin{equation*}
        s_\LR = (1-x_1(e_1))\cdot c_1(e_1) + (1-x_3(e_2))\cdot c_3(e_2) \leq (1-x_1(e_1)) + (1-x_3(e_2)) = x_2(e_1) + x_2(e_2).
    \end{equation*}
    
    In the $\RL$ rounding scheme, we round both items $e_1, e_2$ to agent $2$.
    The total subsidy incurred is
    \begin{align*}
        s_\RL = (1-x_2(e_1)) \cdot a_1 + (1-x_2(e_2)) \cdot a_2 &= a_1 \left((1-x_2(e_1)) + (1-x_2(e_2)) \cdot \alpha\right) \\
        &\leq (1-x_2(e_1)) + (1-x_2(e_2)) \cdot \alpha. \qedhere 
    \end{align*}
\end{proof}

\section{WPROPS Allocation for Goods}

In this section, we consider the allocation of goods and propose an algorithm that computes a WPROPS allocation with total subsidy at most $n/3$.
Since the analysis is almost identical to the ones we have shown in previous sections, we will only highlight the main ideas and changes to the proof, without presenting too much tedious and repetitive analysis.

\subsection{The Notation and Definitions}

We consider the problem of allocating $m$ indivisible goods $M$ to $n$ agents $N$ where each agent $i\in N$ has weight $w_i > 0$ and additive valuation function $v_i:2^M \to \bR^+ \cup \{0\}$.
As before, we assume $\sum_{i\in N} w_i = 1$ and $v_i(e) \leq 1$ for all $i\in N$, $e\in M$.
We define $\WPROP_i$ as agent $i$'s proportional share, i.e., $\WPROP_i = w_i\cdot v_i(M)$.

\begin{definition}[WPROP]
    An allocation $\bX$ is called weighted proportional (WPROP) if $v_i(X_i) \geq \WPROP_i$ for all $i\in N$.
\end{definition}

As before, we use $s_i \geq 0$ to denote the subsidy we give to agent $i\in N$, $\bs = (s_1, \ldots, s_n)$ and $\|\bs\|_1 = \sum_{i\in N} s_i$.

\begin{definition}[WPROPS]
    An allocation $\bX$ with subsidies $\bs = (s_1, \ldots, s_n)$ is called weighted proportional with subsidies (WPROPS) if for any $i\in N$,
    \begin{equation*}
        v_i(X_i) + s_i \geq \WPROP_i.
    \end{equation*}
\end{definition}

Given any allocation $\bX$, computing the minimum subsidy to achieve weighted proportionality can be trivially done by setting
\begin{equation*}
    s_i = \max \{\WPROP_i - v_i(X_i), 0\}, \qquad \forall i\in N.
\end{equation*}

We first show a similar reduction as Lemma~\ref{lemma:reduction-to-IDO} which allows us to consider only IDO instances.
Since the proof is almost identical to that for Lemma~\ref{lemma:reduction-to-IDO}, we omit it.

\begin{definition}[Identical Ordering (IDO) Instances]
    An instance is called identical ordering (IDO) if all agents have the same ordinal preference on the items, i.e., $v_i(e_1) \leq v_i(e_2) \leq \cdots \leq v_i(e_m)$ for all $i\in N$.
\end{definition}

\begin{lemma}
    If there exists a polynomial time algorithm that given any IDO instance computes a PROPS allocation with at most $\alpha$ subsidy, then there exists a polynomial time algorithm that given any instance computes a PROPS allocation with at most $\alpha$ subsidy.
\end{lemma}

With the above reduction, in the following, we only consider the IDO instances. 
Like the chore setting, our algorithm has two main steps: we first compute a fractional PROP allocation, in which a small number of items are fractionally allocated; then we find a way to round the fractional allocation to an integral one. 
Since some agents may have bundle value less than their proportional share after rounding, we offer subsidies to these agents.
By carefully deciding the rounding scheme, we show that the total subsidy required is at most $n/3$.

\subsection{Fractional Bid-and-Take Algorithm}
We use the same notation defined in Section~\ref{ssec:FBTA} for representing a fractional allocation, e.g., we use $x_i(e)\in [0,1]$ to denote the fraction of item $e$ that is allocated to agent $i$.
\paragraph{The Algorithm.}
We continuously allocate the items one by one following an arbitrarily fixed ordering $e_1,e_2,\ldots,e_m$ of the items.
Initially, all agents are active.
For each item $e_j\in M$, we continuously allocate $e_j$ to the active agent $i$ with the maximum $\frac{v_i(e_j)}{v_i(M)}$, until either $e_j$ is fully allocated or $v_i(X_i) = \WPROP_i$. 
If $v_i(X_i) = \WPROP_i$, we inactive agent $i$. 
The algorithm terminates when all items are fully allocated. The steps of the full algorithm are summarized in Algorithm~\ref{alg:FBTA-goods}.

\begin{algorithm}[htbp]
    \caption{Fractional Bid and Take Algorithm}
    \label{alg:FBTA-goods}
    \KwIn{An instance $(M,N,\bw,\mathbf{v})$ with $v_i(e_1) \leq v_i(e_2) \leq \cdots \leq v_i(e_m)$ for all $i\in N$}
    $X_i \gets \mathbf{0}^m, \forall i\in N$  \qquad \qquad \tcp{current fractional bundle}
    $A\gets N$ \qquad \qquad  \qquad \qquad \tcp{ the set of active agents}
    $\mathbf{z} \gets \mathbf{1}^m$ \qquad \qquad  \qquad \qquad\tcp{remaining fraction of the items}
    $j\gets 1$ \qquad \qquad  \qquad \qquad \> \> \tcp{item to be allocated}
    \While{$j \le m$}{
        Let $i \gets \argmax_{i'\in A} \frac{v_{i'}(e_j)}{v_{i'}(M)}$\;
        \If{$v_i(X_i) +z_j\cdot v_i(e_j)>\WPROP_i$}{
            $x_{ie_{j}} \gets \frac{\WPROP_i-v_i(X_i)}{v_i(e_j)}$\;
            $z_j \gets z_j - x_{ie_{j}}$, $A \gets A\setminus \{i\}$\;
            \If{$|A| = 1$}{
                Allocate all remaining items to the only active agent and go to output\;
            }
        }
        \Else{
            $x_{ie_{j}} \gets z_{j}$, $z_{j} \gets 0$, $j\gets j+1$\;
        }    
    }
    \KwOut{A fractional allocation $\bx$.}
\end{algorithm}

\begin{lemma}\label{lemma: wprop-goods}
    The output allocation $\bx$ is a fractional WPROP allocation. 
\end{lemma}
\begin{proof}
    It suffices to show that the last active agent $n$ receives a bundle of value at least $\WPROP_n$. 
    Since we allocate each item $e$ to the active agent $i$ with the maximum $\frac{v_i(e)}{v_i(M)}$, we have $\frac{v_i(e)}{v_i(M)} \geq \frac{v_j(e)}{v_j(M)}$ for every active $j\neq i$. Therefore we have $\frac{v_j(X_i)}{v_j(M)} \leq \frac{v_i(X_i)}{v_i(M)}$ for any active agents $i,j\in A$. Suppose that $v_n(X_n) < \WPROP_n$ at the end of the algorithm, then we have 
    \begin{equation*}
        1 = \sum_{i\in N} \frac{v_n(X_i)}{v_n(M)}
        \le \sum_{i\in N}\frac{v_i(X_i)}{v_i(M)} < \sum_{i\in N}\frac{\WPROP_i}{v_i(M)} = \sum_{i\in N} w_i = 1,
    \end{equation*}
    which is a contradiction.
\end{proof}

The definition of item-sharing graph can be naturally extended to the goods case.
We call agent $j$ a successor of agent $i$ if Algorithm~\ref{alg:FBTA-goods} tries to allocate the remaining part of item $e^i$ to $j$.
By letting the set of nodes be the agents where each agent has a directed edge to its successor (if any), we construct the item-sharing graph of the fractional allocation $\bx$ returned by Algorithm~\ref{alg:FBTA-goods}.
As before, we differ in two cases of whether there exist \emph{shattered items}, i.e., items shared by three or more agents.
We remark that the splitting we introduced in the chores setting can also be used in the goods setting since it only deals with the tree structure instead of the cost/value of items.

\subsection{Graph without Shattered Item}
We first consider the case that there is no shattered item in the item-sharing graph $G$.
Following the same analysis as in the proof of Lemma~\ref{thm:three-agents}, we can show that for any tree with size two, there exists a rounding scheme that guarantees a WPROPS allocation with a total subsidy of at most $2/3$.

\begin{theorem}\label{thm:three-agents-goods}
    For a tree with size $2$, there exists a rounding scheme that guarantees a WPROPS allocation with total subsidy at most $2/3$, which is optimal.
\end{theorem}

Consider an instance $\cI = (M,N,\bw,\bv)$ where $N = \{1,2,3\}$ (see Figure~\ref{fig:example-three-agents-goods}).
Let $a_1 = v_2(e_1)$, $a_2 = v_2(e_2)$, we assume w.l.o.g. that $a_1 \geq a_2$.
We let $\alpha = a_1 / a_2$.

    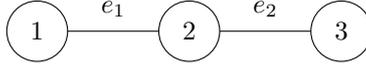
\begin{figure}[ht]
        \centering
        \begin{tikzpicture}	
            \draw (-2,0) circle(0.4);    \node at (-2,0) {$1$};
            \draw (-1.6,0)--(-0.4,0);   \node at(-1,0.3) {$e_1$};
            \draw (0,0) circle(0.4);    \node at (0,0) {$2$};
            \draw (0.4,0)--(1.6,0);     \node at(1,0.3) {$e_2$};
            \draw (2,0) circle(0.4);    \node at (2,0) {$3$};
        \end{tikzpicture}
        \caption{The item-sharing graph for three agents case.}
        \label{fig:example-three-agents-goods}
    \end{figure}
    
Similar to the chore setting, we introduce four different rounding schemes called $\LL,\RR,\LR,\RL$.
We use $s_{\LL}, s_{\RR}, s_{\LR}, s_{\RL}$ to denote the subsidy required in the four rounding strategies respectively.
The upper bounds for the subsidies required for the four different rounding schemes are as shown in Table~\ref{tab:four-upperbounds-goods}.

    \begin{table}[htbp]
        \centering
        \begin{tabular}{|c|c|}
        \hline
        \textbf{Rounding Schemes} & \textbf{Upper Bounds for Subsidies}   \\ \hline
        $\LL$ & $(1-x_2(e_2)) + \left(x_2(e_1) - x_2(e_2)\cdot\alpha\right)^+$ \\ \hline
        $\RR$ & $(1-x_2(e_1)) + \left(x_2(e_2)\cdot\alpha - x_2(e_1)\right)^+$ \\ \hline
        $\LR$ & $x_2(e_1) + x_2(e_2)\cdot \alpha$ \\ \hline
        $\RL$ & $(1-x_2(e_1)) + (1-x_2(e_2))$          \\ \hline
        \end{tabular}
        \caption{Upper bounds for the subsidies required for the four different rounding schemes.}\label{tab:four-upperbounds-goods}
    \end{table}

If we use $y_1$ and $y_2$ to denote $1-x_2(e_1)$ and $1-x_2(e_2)$, then the above upper bounds are symmetrical as the upper bounds we present in the proof of Theorem~\ref{thm:three-agents}. 
Note that $0 \leq 1 - x_2(e_i) \leq 1$ for both $i\in \{1,2\}$, the proof of Theorem~\ref{thm:three-agents} can be extended to the goods case straightforwardly.

Based on Theorem~\ref{thm:three-agents-goods}, we call \splitting
to decompose any graph $G$ (without shattered item) into a collection of subtrees with size two.
By Lemma~\ref{lemma:simple-splitting} and the same analysis as in the proof of Theorem~\ref{thm:subsidy-n-1}, we can bound the total subsidy required for the graph without shattered item by $n/3$.

\begin{theorem}\label{thm:subsidy-n-1-goods}
    Given the fractional allocation $\bx$ with $n-1$ fractional items returned by Algorithm~\ref{alg:FBTA-goods}, there exists a rounding scheme that returns an integral WPROPS allocation $\bX$ with total subsidy at most $n/3$.
\end{theorem}

\subsection{Graph with Shattered Item}
In this section, we consider the case that there exists at least one shattered item in the item-sharing graph $G$.
Similar to the chores setting, we call Algorithm~\ref{alg:GD} to decompose the graph $G$ into an expanded atom-path and a collection of trees.
Following Lemma~\ref{lemma:atom-path-splitting}, each decomposed tree returned by the Atom-path Splitting either contains an atom-path or has an even size.
To upper bound the total required subsidy, it suffices to bound the subsidy required for each expanded atom-path.

\begin{lemma}\label{lemma:expanded-atom-path-goods}
    For any expanded atom-path with $k$-length atom-path and $h\leq k+1$ attached edges, there exists a rounding scheme with total subsidy bounded by $(k+h)/3$.
\end{lemma}

Given Lemma~\ref{lemma:expanded-atom-path-goods}, we are ready to upper bound the total subsidy required for graph $G$.
The proof of Lemma~\ref{lemma:general-tree-subsidy-goods} is based on mathematical induction and similar to the proof of Lemma~\ref{lemma:general-tree-subsidy}.
Here we omit the detailed proof.

\begin{lemma}\label{lemma:general-tree-subsidy-goods}
    For any tree with size $z \geq 2$ such that there exists an atom-path in the tree, there exists a rounding scheme with total subsidy bounded by $z/3$.
\end{lemma}

We remark that Lemma~\ref{lemma:general-tree-subsidy-goods} directly leads to the following theorem.

\begin{theorem}\label{thm:subsidy-<n-1-goods}
    Given the fractional allocation $\bx$ with less than $n-1$ fractional items returned by Algorithm~\ref{alg:FBTA}, there exists a rounding scheme that returns an integral WPROPS allocation with total subsidy at most $(n-1)/3$.
\end{theorem}

Next, we define some notations that will be useful to prove Lemma~\ref{lemma:expanded-atom-path-goods}.
Fix any expanded atom-path $T$, we assume that the atom-path has $k\geq 2$ edges, which correspond to the item $e_0$.
We denote the nodes in the atom-path by agents $1, \ldots, k+1$, where agent $i+1$ is the successor of agent $i$, for all $i=1,2,\ldots,k$.
We denote the edge attached to agent $i$ by $e_i$.
When there is no attached edge to agent $i$, we define $e_i = \bot$.
We call agents $1,2,\ldots,k+1$ the \emph{atom-path agents} and the other agents the \emph{attached agents}.
We use $h \leq k+1$ to denote the number of attached edges/agents in the expanded atom-path.
For each item $e_i$, where $i\in \{1,\ldots,k+1\}$, we use $y_i$ to denote the fraction of item $e_i$ held by the corresponding attached agent, i.e., $y_i = 1 - x_i(e_i)$.
For all $i\in \{1, \ldots, k+1\}$, when there is no ambiguity, we use $x_i$ to denote $x_i(e_0)$ for convenience.

In the following, we focus on upper bounding the subsidy required for any expanded atom-path.
We show Lemma~\ref{lemma:expanded-atom-path-goods}, by giving the following arguments of different types of expanded atom-paths respectively, and show that the arguments cover all types of expanded atom-paths.

\begin{lemma}\label{lemma:expanded-atom-path-general-goods}
    For an expanded atom-path with $h \leq k \cdot (2-\frac{6}{k+1})$, there exists a rounding scheme with total subsidy at most $(k+h)/3$.
\end{lemma}

\begin{lemma}\label{lemma:expanded-atom-path-k+1-goods}
    For an expanded atom-path with $h = k+1$, there exists a rounding scheme with total subsidy at most $(k+h)/3$.
\end{lemma}

\begin{lemma}\label{lemma:expanded-atom-path-k-goods}
    For an expanded atom-path with $h = k$, there exists a rounding scheme with total subsidy at most $(k+h)/3$.
\end{lemma}

\begin{lemma}\label{lemma:expanded-atom-path-k-1-goods}
    For an expanded atom-path with $2 \leq k \leq 3$ and $h = k-1$, there exists a rounding scheme with total subsidy at most $(k+h)/3$.
\end{lemma}

Similar to the proof of Lemma~\ref{lemma:expanded-atom-path}, the above arguments immediately lead to Lemma~\ref{lemma:expanded-atom-path-goods}.
Hence we omit the proof here.

\subsubsection{Expanded Atom-path}
\begin{replemma}{lemma:expanded-atom-path-general-goods}
    For an expanded atom-path with $h \leq k \cdot (2-\frac{6}{k+1})$, there exists a rounding scheme with total subsidy at most $(k+h)/3$.
\end{replemma}
\begin{proof}
    We consider the threshold rounding that rounds each item to the agent holding the maximum fraction of it.
    For item $e_0$, we round it to $i^* = \argmax_{1 \leq i \leq k+1} x_i$, which incurs a subsidy of at most $\frac{k}{k+1}$.
    The rounding of each item $e_i$ incurs subsidy at most $\frac{1}{2}$ iff $e_i \neq \bot$.
    Hence the total subsidy is bounded by $\frac{k}{k+1} + \frac{h}{2}$, which is at most $\frac{k+h}{3}$ when $h \leq k\cdot (2- \frac{6}{k+1})$.
\end{proof}

We remark that the major difference between the analysis of goods and chores is the biased threshold rounding.
In the case of goods, rounding item $e_0$ to agent $i^*$ would not release a space of $x_i\cdot v_i(e)$ to agent $i\neq i^*$.
Instead, it creates a gap of $x_i\cdot v_i(e)$ to the proportionality and needs to be filled by rounding other items or subsidizing money.
Next, we show that in the goods setting, there also exists a biased threshold rounding that makes use of the dependent rounding of items.

For any agent $i\neq i^*$ that does not receive item $e_0$, when we round $e_i$ to agent $i$, the total subsidy required for agent $i$ and her (corresponding) attached agent is at most $(x_i \cdot v_i(e_0) - y_i \cdot v_i(e_i))^+ + y_i$, while the total subsidy required for agent $i$ in the case that rounds $e_i$ to the attached agent is $x_i \cdot v_i(e_0) + (1-y_i) \cdot v_i(e_i)$.
For any $1 \leq i \leq k+1$, we use $\alpha_i$ to denote $v_i(e_0)$.
We formalize this idea into the following rounding scheme.

\paragraph{Biased Threshold Rounding (Goods).}
For any attached edge $e_i$, we round item $e_i$ to the attached agent if $x_i\cdot \alpha_i + (1-y_i) \cdot v_i(e_i) < (x_i \cdot \alpha_i - y_i \cdot v_i(e_i))^+ + y_i$; otherwise round item $e_i$ to agent $i$.

\medskip

\begin{proposition}\label{proposition:modified-greedy-goods}
    For any agent $i\in \{1,2,\ldots,k+1\}$ that does not receive item $e_0$, by rounding $e_i$ using the biased threshold rounding, the subsidy required for agent $i$ and her corresponding attached agent is at most $(1 + x_i\cdot \alpha_i)/2$ if $x_i \cdot \alpha_i \leq \frac{1}{2}$; at most $x_i \cdot \alpha_i + 1/4$ if $x_i \cdot \alpha_i > \frac{1}{2}$.
\end{proposition}
\begin{proof}
    Note that the incurred subsidy by biased threshold rounding is bounded by
    \begin{equation*}
        \min \{(x_i \cdot \alpha_i - y_i \cdot v_i(e_i))^+ + y_i,\; x_i\cdot \alpha_i + (1-y_i) \cdot v_i(e_i)\}.
    \end{equation*}
    
    When $(x_i \cdot \alpha_i - y_i \cdot v_i(e_i))^+ = 0$, the total subsidy required for agent $i$ and her corresponding attached agent is at most 
    \begin{equation*}
        \min \{y_i,\; x_i \cdot \alpha_i + (1-y_i) \cdot v_i(e_i)\} \leq \min \{y_i,\; x_i \cdot \alpha_i + 1-y_i\} \leq \frac{1 + x_i \cdot \alpha_i}{2}.
    \end{equation*}
    which is not larger than $x_i \cdot \alpha_i + 1/4$ when $x_i \cdot \alpha_i > \frac{1}{2}$.
        
    By Proposition~\ref{proposition:xy<1/4}, when $(x_i \cdot \alpha_i - y_i \cdot v_i(e_i))^+ > 0$, the total subsidy required for agent $i$ and her corresponding attached agent is at most
    \begin{equation*}
        \min \{x_i \cdot \alpha_i + y_i \cdot (1-v_i(e_i)),\; x_i \cdot \alpha_i + (1-y_i) \cdot v_i(e_i)\} \leq x_i \cdot \alpha_i+ \frac{1}{4}  \qedhere
    \end{equation*}
    which is not larger than $(1+x_i \cdot \alpha_i)/2$ when $x_i \cdot \alpha_i \leq \frac{1}{2}$.
\end{proof}

\begin{corollary}\label{corollary:modified-greedy-goods}
    For any agent $i\in \{1,2,\ldots,k+1\}$ that does not receive item $e_0$, by rounding $e_i$ using the biased threshold rounding, the subsidy required for agent $i$ and her corresponding attached agent is at most $(1 + x_i)/2$ if $x_i \cdot \alpha_i \leq \frac{1}{2}$.
\end{corollary}

Given Proposition~\ref{proposition:modified-greedy-goods}, we show that the proofs of Lemma~\ref{lemma:expanded-atom-path-general-goods}, \ref{lemma:expanded-atom-path-k+1-goods}, \ref{lemma:expanded-atom-path-k-goods}, \ref{lemma:expanded-atom-path-k-1-goods} are similar to the proofs of analogous lemmas in the chores allocation with minor modification.

\begin{observation}\label{obser:v_i(e_0)>=v_i(e_i)}
    For any agent $i \leq k$, $v_i(e_0) \geq v_i(e_i)$.
\end{observation}
\begin{proof}
    We remark that for any agent $i \leq k$, $e_0$ is the last (fraction of) item that she receives.
    Otherwise, agent $i$ would not have a successor with respect to $e_0$.
    In other words, for any agent $i$, the event that $i$ takes (a fraction of) item $e_i$ happens before the event that $i$ takes (a fraction of) item $e_0$.
    Recall that we reduce all instances to the IDO instance with increasing value for all agents, and use the IDO instance as the input of Algorithm~\ref{alg:FBTA-goods},
    Hence, for any agent $i \leq k$, we have $v_i(e_0) \geq v_i(e_i)$.
\end{proof}

\begin{proposition}\label{proposition:x_j>1/2}
    For an expanded atom-path with $k-1 \leq h \leq k+1$ attached agents, if there exists an agent $1 \leq j\leq k+1$ such that $x_j \cdot \alpha_j > \frac{1}{2}$, there exists a rounding scheme with total subsidy at most $(k+h)/3$.
\end{proposition}
\begin{proof}
    If there exists an agent $1 \leq j\leq k+1$ such that $x_j > \frac{1}{2}$, we apply the threshold rounding to all fractional items.
    The rounding of item $e_0$ incurs a subsidy of at most $1/2$ since $x_j > \frac{1}{2}$ and the rounding of any item $e_i$ such that $e_i \neq \bot$ incurs a subsidy of at most $1/2$.
    Hence the total required subsidy is at most
    \begin{equation*}
        \|s\|_1 \leq \frac{h+1}{2},
    \end{equation*}
    which is not larger than $(k+h)/3$ when $k \geq 4, h = k+1$, or $k\geq 3, h = k$, or $k\geq 2, h = k-1$.
    Hence it remains to consider the cases of $k = 3, h = 4$, or $k = 2, h = 3$, or $k = 2, h = 2$.

    In the following, we ignore the directions of the attached edges and assume w.l.o.g. that $x_{k+1} \cdot \alpha_{k+1} > \frac{1}{2}$.
    We first consider the case that $k = 3, h  = 4$, where $(k+h)/3 = 7/3$.
    Consider the rounding scheme that rounds item $e_0$ to agent $4$, rounds items $e_1, e_2, e_3$ using the biased threshold rounding, and rounds item $e_4$ using the threshold rounding.
    The total required subsidy can be bounded by
    \begin{equation*}
        \sum_{1 \leq i \leq 3} \frac{1+x_i}{2} + \frac{1}{2} = 2 + \frac{x_1 + x_2 + x_3}{2} < \frac{9}{4} <\frac{7}{3}.
    \end{equation*}

    We move to consider the case that $k = 2, h = 3$, where $(k+h)/3 = 5/3$.
    We consider the rounding scheme that rounds item $e_0$ to agent $1$, rounds item $e_1, e_2$ using the biased threshold rounding, and rounds item $e_3$ to the attached agent.
    The total required subsidy is at most
    \begin{align*}
        ((1-y_3) - (1-x_3)\cdot \alpha_3)^+ + \frac{1+x_1}{2} + \frac{1+x_2}{2}.
    \end{align*}
    When $((1-y_3) - (1-x_3)\cdot \alpha_3)^+ = 0$, the upper bound becomes
    \begin{equation*}
        \frac{1+x_1}{2} + \frac{1+x_2}{2} < \frac{5}{4}.
    \end{equation*}
    When $((1-y_3)- (1-x_3)\cdot \alpha)^+ > 0$, the upper bound becomes
    \begin{equation*}
        (1-y_3) - (1-x_3)\cdot \alpha_3 + \frac{1+x_1}{2} + \frac{1+x_2}{2}.
    \end{equation*}
    Consider another rounding scheme that rounds item $e_0$ to agent $4$, rounds items $e_1, e_2$ using the biased threshold rounding, and rounds item $e_4$ to agent $4$.
    The total required subsidy is at most.
    \begin{equation*}
        y_3 + \frac{1+x_1}{2} + \frac{1+x_2}{2}
    \end{equation*}
    The summation of the above two upper bounds is at most
    \begin{equation*}
        1 - (1-x_3) \cdot \alpha_3 + 2 + x_1 + x_2 = 3 + (1-x_3) - (1-x_3) \cdot \alpha_3 = 3 + (1-x_3) \cdot (1- \alpha_3) < \frac{13}{4},
    \end{equation*}
    where the last inequality holds following Proportion~\ref{proposition:xy<1/4} and $x_3 \cdot \alpha_3 > 1/2$.
    Hence, at least one of the two rounding schemes requires a total subsidy of at most $13/8 < 5/3$.
    
    Finally, we consider the case that $k = 2, h = 2$, where $(k+h)/3 = 4/3$.
    We assume w.l.o.g. that $x_3 \cdot \alpha_3 > \frac{1}{2}$ and $x_2 = \bot$.
    We consider the rounding scheme that rounds item $e_0$ to agent $3$, rounds item $e_1$ using the biased threshold rounding, and rounds item $e_3$ to the attached agent.
    The total required subsidy is at most
    \begin{equation*}
        \frac{1+x_1}{2} + x_2 + ((1-y_3)- (1-x_3)\cdot \alpha_3)^+
    \end{equation*}
    When $((1-y_3) - (1-x_3)\cdot \alpha_3)^+ = 0$, the upper bound becomes
    \begin{equation*}
        \frac{1+x_1}{2} + x_2 <\frac{1+x_1 + x_2 + x_3}{2} = 1.
    \end{equation*}
    When $((1-y_3)- (1-x_3)\cdot \alpha_3)^+ > 0$, the upper bound becomes
    \begin{equation*}
        \frac{1+x_1}{2} + x_2 + (1-y_3)- (1-x_3)\cdot \alpha_3.
    \end{equation*}
    Consider another rounding scheme that rounds item $e_0$ to agent $3$, rounds item $e_1$ using the biased threshold rounding, and rounds item $e_3$ to agent $3$.
    The total required subsidy is at most
    \begin{equation*}
        \frac{1+x_1}{2} + x_2 + y_3.
    \end{equation*}
    The summation of the above two upper bounds is at most
    \begin{align*}
        & 1 + x_1 + 2 \cdot x_2 + 1 - (1-x_3) \cdot \alpha_3\\
        = &2 + (1-x_3) + x_2 - (1-x_3) \cdot \alpha_3  \\
        = &2 + (1-x_3) \cdot (1- \alpha_3) + x_2< \frac{9}{4} + x_2,
    \end{align*}
    where the last inequality holds following Proportion~\ref{proposition:xy<1/4} and $x_3 \cdot \alpha_3 > 1/2$.
    When $x_2 \leq \frac{5}{12}$ the above summation is at most $8/3$, leading to at least one of the rounding schemes requiring total subsidy at most $4/3$.
    Consider otherwise that $x_2 > 5/12$, which implies that $x_1 + x_3 < 7/12$.
    We consider the rounding scheme that rounds item $e_0$ to agent $2$, rounds items $e_1, e_3$ using the biased threshold rounding, the total subsidy is at most
    \begin{equation*}
        \frac{1+x_1}{2} + x_3 + \frac{1}{4} \leq x_1 + x_3 + \frac{3}{4} < \frac{4}{3}.
    \end{equation*}
    
    Hence for all expanded atom-path with $x_j \cdot \alpha_j > \frac{1}{2}$ for some $1 \leq j\leq k+1$, there exists a rounding scheme with total subsidy at most $(k+h)/3$.
\end{proof}

By Proposition~\ref{proposition:x_j>1/2}, in the following, we only consider cases that $x_i \cdot \alpha_i \leq 1/2$ for all $1 \leq i \leq k+1$.
In other words, we assume that Corollary~\ref{corollary:modified-greedy-goods} holds for any $1 \leq i \leq k+1$ that does not receive item $e_0$.

\begin{replemma}{lemma:expanded-atom-path-k+1-goods}
    For an expanded atom-path with $h = k+1$, there exists a rounding scheme with total subsidy at most $(k+h)/3$.
\end{replemma}

\begin{proof}
    We consider the rounding schemes that round $e_0$ to agent $1$ and round the items $e_2, \ldots, e_{k+1}$ using the biased threshold rounding.
    We show that this rounding scheme has a total subsidy of no more than $(k+h)/3$.
    \begin{itemize}
        \item \textbf{Scheme 1.} We round item $e_1$ to the attached agent.
        \item \textbf{Scheme 2.} We round item $e_1$ to agent $1$.
    \end{itemize}
    
    For scheme 1, the total required subsidy can be bounded by the following
    \begin{align}
        &((1- y_1) \cdot v_1(e_1) - (1-x_1) \cdot v_1(e_0))^+ + \sum_{2 \leq i \leq k+1} \frac{1 + x_i}{2} \notag\\
        \leq & ((x_1 - y_1) \cdot v_1(e_0))^+ + \sum_{2 \leq i \leq k+1} \frac{1 + x_i}{2}   \tag{$v_1(e_1) \leq v_1(e_0)$}\\
        \leq &  (x_1 - y_1)^+ + \sum_{2 \leq i \leq k+1} \frac{1 + x_i}{2}. \label{eq:k+1-scheme1-goods}
    \end{align}
    
    For scheme 2, the total required subsidy can be bounded by
    \begin{equation}
        y_1 + \sum_{2 \leq i \leq k+1} \frac{1 + x_i}{2}.   \label{eq:k+1-scheme2-goods}
    \end{equation}
    
    When $x_1 - y_1 \leq 0$, it is easy to find that Equation~\eqref{eq:k+1-scheme1-goods} dominates Equation~\eqref{eq:k+1-scheme2-goods}.
    We can upper bound the subsidy for scheme 1 as follows.
    \begin{align*}
        \|s\|_1 &\leq \sum_{2 \leq i \leq k+1} \frac{1 + x_i}{2}  &\tag{By Equation~\eqref{eq:k+1-scheme1} and $x_1 - y_1 \leq 0$}\\
        & = \frac{k}{2} + \sum_{2 \leq i \leq k+1} \frac{x_i}{2} \leq \frac{k+1}{2}.  &\tag{$\sum_{2 \leq i \leq k+1} x_i \leq 1$}
    \end{align*}
    
    Next, we consider the case that $x_1 - y_1 > 0$.
    We have that
    \begin{align*}
        \text{For scheme 1, } \|s\|_1 &\leq x_1 - y_1 + \sum_{2 \leq i \leq k+1} \frac{1 + x_i}{2};  &\tag{By Equation~\ref{eq:k+1-scheme1-goods}}\\
        \text{For scheme 2, } \|s\|_1 &\leq y_1 + \sum_{2 \leq i \leq k+1} \frac{1 + x_i}{2}.   &\tag{By Equation~\ref{eq:k+1-scheme2-goods}}
    \end{align*}
    
    Summing up the above two upper bounds, We have that the total subsidy required by the two rounding schemes combined is at most
    \begin{align*}
        &x_1 - y_1 + y_i + \sum_{2 \leq i \leq k+1} \frac{1 + x_i}{2} + \sum_{2 \leq i \leq k+1} \frac{1 + x_i}{2} \\
        = &x_1 + \sum_{2 \leq i \leq k+1} (1 + x_i)   \\
        = &k+1.  &\tag{$\sum_{1 \leq i \leq {k+1}} x_i = 1$}
    \end{align*}
    
    Therefore, at least one of the two rounding schemes requires a total subsidy of at most $(k+1)/2$, which is not larger than $(k+h)/3$ for all $k\geq 2$.
\end{proof}

\begin{replemma}{lemma:expanded-atom-path-k-goods}
    For an expanded atom-path with $h = k$, there exists a rounding scheme with total subsidy at most $(k+h)/3$.
\end{replemma}

\begin{proof}
    When $h = k$, there exists an agent $j$ in the atom-path with no edge attached, i.e., $e_j = \bot$.
    We round item $e_0$ to agent $j$ and round all item $e_i$ s.t. $1 \leq i\leq k, i\neq j$ using the biased threshold rounding.
    \begin{equation*}
        \sum_{i \neq j} \frac{1 + x_i}{2} \leq \frac{k}{2} + \sum_{i\neq j}\frac{x_i}{2} = \frac{k + 1 - x_j}{2},
    \end{equation*}
    which is not larger than $(k+h)/3$ when $k\geq 3$; or $k = 2$ and $x_j \geq 1/3$.
    It remains to consider the case that $k = 2$ and $x_j < 1/3$.
    By Observation~\ref{obser:v_i(e_0)>=v_i(e_i)}, there exists (at least) one agent $i\in N \setminus \{j\}$ such that $v_{i}(e_0) \geq v_{i}(e_{i})$.
    Let $p$ be the agent other than $i,j$.
    We round item $e_0$ to agent $i$ and round items $e_p$ using the biased threshold rounding.
    We show that at least one of the following rounding schemes gives a total subsidy of at most $4/3$.
    \begin{itemize}
        \item \textbf{Scheme 1.} We round item $e_i$ to the attached agent.
        \item \textbf{Scheme 2.} We round item $e_i$ to agent $i$.
    \end{itemize}
    
    By a similar argument as in Lemma~\ref{lemma:expanded-atom-path-k+1-goods}, the total subsidy required for scheme 1 is at most
    \begin{equation}
        (x_i - y_i)^+ + x_j + \frac{1+x_p}{2}. \label{eq:k-scheme1-goods}
    \end{equation}
    
    For scheme 2, the total required subsidy can be bounded by
    \begin{equation}
        y_i + x_j + \frac{1+x_p}{2}.  \label{eq:k-scheme2-goods}
    \end{equation}
    
    When $x_i - y_i \geq 0$, the sum of the above two upper bounds is
    \begin{equation*}
        x_i +  2x_j + (1 + x_p) \leq 2 + x_j < \frac{7}{3}.
    \end{equation*}
    
    Hence at least one of the two rounding schemes is with total subsidy bounded by $7/6 < 4/3$.
    
    When $x_i - y_i < 0$, Equation~\eqref{eq:k-scheme1-goods} becomes
    \begin{equation*}
        x_j + \frac{1+x_p}{2}.
    \end{equation*}
    
    Consider another rounding scheme that rounds item $e_0$ to agent $j$, rounds item $e_i$ using the threshold rounding, and rounds item $e_p$ using the biased threshold rounding.
    The total subsidy is at most
    \begin{equation*}
        x_i + \frac{1}{2} + \frac{1+x_p}{2}.
    \end{equation*}
    
    Summing up the above two upper bounds leads to
    \begin{equation*}
        x_i + x_j + x_p + \frac{3}{2} = \frac{5}{2}.
    \end{equation*}
    
    Hence at least one of the rounding schemes requires subsidy at most $\frac{5}{4} < \frac{4}{3}$.
\end{proof}

\begin{replemma}{lemma:expanded-atom-path-k-1-goods}
    For an expanded atom-path with $2 \leq k \leq 3$ and $h = k-1$, there exists a rounding scheme with total subsidy at most $(k+h)/3$.
\end{replemma}

\begin{proof}
    For $k = 3$, the upper bound is $\frac{2k-1}{3} = \frac{5}{3}$.
    Note that there exists two agents $i,j$ such that ${e_i} = {e_j} = \bot$.
    Let $p,q$ be the two agents other than $i,j$.
    Consider the rounding scheme that rounds $e_0$ to $\argmax \{x_i,x_j\}$, say agent $i$.
    Applying the biased threshold rounding to items $e_p$ and $e_q$, we can upper bound the total required subsidy by 
    \begin{equation*}
        \frac{1+x_p}{2} + \frac{1+x_q}{2} + x_j = 1 + \frac{x_p + x_q + 2\cdot x_j}{2} \leq 1 + \frac{x_p + x_q + x_j+ x_i}{2} = \frac{3}{2}.
    \end{equation*}
        
    For $k=2$, the upper bound is $\frac{2k-1}{3} = 1$.
    Note that there exist two agents $i,j$ such that ${e_i} = {e_j} = \bot$.
    Let $p$ be the agent other than $i,j$ (who has an attached edge).
    We round $e_0$ to $\argmax \{x_i,x_j\}$, say agent $i$, and apply the biased threshold rounding to items $e_p$.
    The total required subsidy is at most
    \begin{equation*}
        \frac{1+x_p}{2} + x_j = \frac{1+x_p + 2\cdot x_j}{2} \leq \frac{1+x_p + x_j+ x_i}{2} = 1.   \qedhere
    \end{equation*}
\end{proof}

\end{document}